\documentclass{article}

\usepackage[final,main]{neurips_2026}

\usepackage[utf8]{inputenc}
\usepackage[T1]{fontenc}
\usepackage{hyperref}
\usepackage{url}
\usepackage{booktabs}
\usepackage{comment}
\usepackage{amsfonts}
\usepackage{amsmath}
\usepackage{amssymb}
\usepackage{nicefrac}
\usepackage{tikz}
\usetikzlibrary{arrows.meta, decorations.markings,
                   decorations.pathmorphing, positioning, calc}
\usepackage{xcolor}
\usepackage{bm}
\usepackage{mathtools}
\usepackage{amsthm}
\usepackage{multirow}
\newtheorem{theorem}{Theorem}
\newtheorem{corollary}[theorem]{Corollary}
\newtheorem{lemma}[theorem]{Lemma}
\usepackage{float}
\usepackage{makecell}

\colorlet{revcolor}{black}
\newcommand{\rev}[1]{{\color{revcolor}\relax#1}}

\newcommand{\klw}{{\mathrm{KL}_{\text{W}}}}
\newcommand{\dvr}{{R_{\mathrm{DV}}}}
\newcommand{\KL}{D_{\mathrm{KL}}}
\newcommand{\E}{\mathbb{E}}
\newcommand{\R}{\mathbb{R}}
\newcommand{\ket}[1]{|#1\rangle}
\newcommand{\bra}[1]{\langle #1|}
\newcommand{\braket}[2]{\langle #1|#2\rangle}
\newcommand{\expect}[1]{\langle #1 \rangle}
\newcommand{\dW}{\mathrm{d}W}
\newcommand{\dt}{\mathrm{d}t}

\title{\rev{QMaxCal: Decoherence-Aware Quantum Optimal Control via Girsanov's Theorem}}

\author{%
  Merijn Moody$^{1,2,3,*}$ \enspace
  Zier Mensch$^{2,4,*}$ \\
  \bfseries Miranda C.N.\ Cheng$^{2,3,5}$ \enspace
  Peter G.\ Bolhuis$^{6}$ \enspace
  Max Welling$^{7}$ \\[4pt]
  $^{1}$Dutch Institute for Emergent Phenomena, University of Amsterdam, Netherlands \\
  $^{2}$Institute of Physics, University of Amsterdam, Netherlands \\
  $^{3}$Korteweg-de Vries Institute for Mathematics, University of Amsterdam, Netherlands \\
  $^{4}$Department of Physics, National Taiwan University, Taiwan \\
  $^{5}$Institute for Mathematics, Academia Sinica, Taiwan \\
  $^{6}$Van 't Hoff Institute for Molecular Sciences, University of Amsterdam, Netherlands \\
  $^{7}$Amsterdam Machine Learning Lab, University of Amsterdam, Netherlands \\[4pt]
  $^{*}$Equal contribution. Correspondence to: \texttt{merijnmoody@gmail.com}
}

\begin{document}

\maketitle

\begin{abstract}
Reliable quantum control in the presence of decoherence requires policies that combat the effect of environmental noise on the controlled dynamics. Open quantum systems under continuous monitoring generate classical measurement records whose drift depends on the noise experienced by the system; the records of two evolutions sharing the same decoherence channels differ only in this drift, so Girsanov's theorem yields a closed-form, differentiable estimator of the KL divergence between their trajectory distributions. We instantiate this estimator with two physically motivated reference measures, yielding two regularizers: the \emph{Wiener KL} $(\mathrm{KL}_{\text{W}})$, which vanishes on the joint kernel of the Lindblad operators and is the appropriate choice when this kernel is reachable by control, and the \emph{drift-variance regularizer} $(R_{\mathrm{DV}})$, which, for dephasing and decay channels, vanishes if and only if the evolution is decoherence-free. Both are qualitatively distinct from existing penalties on control fluence or smoothness: they penalize the observable consequences of control on the decoherence channels rather than the control amplitude itself. The regularizers outperform unregularized gradient-based and reinforcement-learning baselines across a range of open quantum systems, including single- and multi-qubit benchmarks and a six-qubit chain calibrated to a published snapshot of the IBM Kingston processor. We evaluate them on final-state fidelity, robustness to mismatch in the assumed noise model (gains grow from $+17$\,pp at training noise to $+27$\,pp under $2.5\times$ noise mismatch), and occupation of forbidden states. The regularizers reduce infidelity by up to 50\%, with $\sim$16\% gains on the calibrated IBM Kingston chain. Our code can be found at \url{https://github.com/MerijnMoody/QMaxCal}.
\end{abstract}

\section{Introduction}
\label{sec:intro}

\rev{
\begin{figure}[t]
    \centering
    \includegraphics[width=\linewidth]{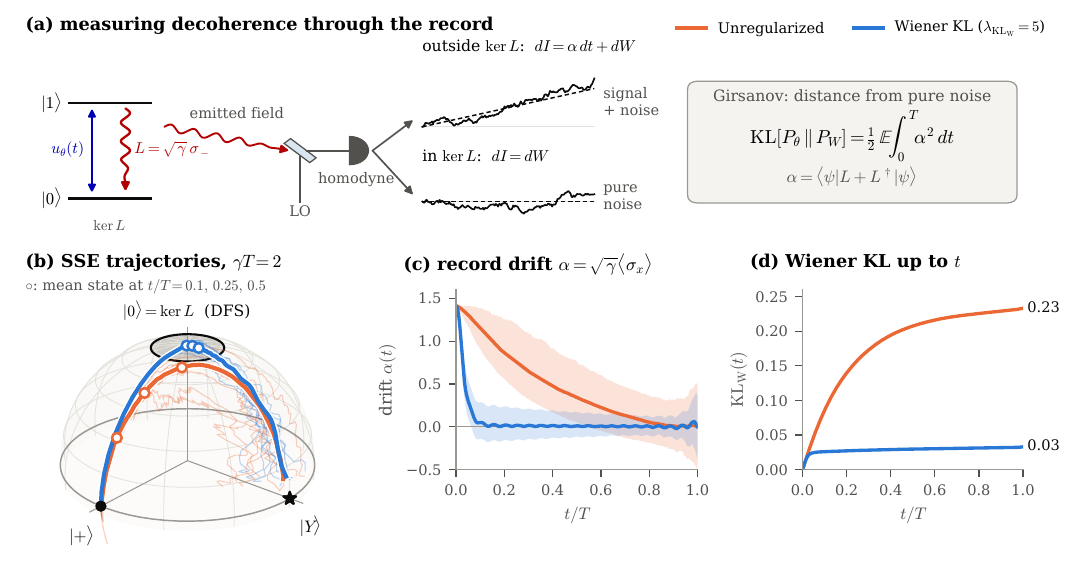}
    \caption{\textbf{Measuring decoherence through the measurement record.} \textbf{(a)} Homodyne detection of the emitted field gives a record $dI = \alpha\,dt + \dW$ with drift $\alpha = \bra{\psi}L + L^\dagger\ket{\psi}$. In $\ker L$ the record is pure noise, with law $P_W$; the Wiener KL measures how far the law $P_\theta$ of the records is from $P_W$ (the two records are illustrations). \textbf{(b)} SSE trajectories of the amplitude-damped qubit, $\ket{+} \to \ket{Y}$ at $\gamma T = 2$ (thick lines: mean Bloch vector). The Wiener-KL policy ($\lambda_{\klw} = 5$) reaches $\ker L$ by $t = 0.1T$ and stays there until the final rotation. \textbf{(c)} Record drift, mean and 10th to 90th percentile. \textbf{(d)} Wiener KL accumulated up to time $t$; it stops growing once the state is in $\ker L$.}
    \label{fig:method}
\end{figure}
}

Quantum optimal control provides control policies that underlie a wide range of tasks in quantum technology, from implementing high-fidelity gates in superconducting and trapped-ion processors \citep{werninghaus2021,motzoi2009} to preparing non-classical states for sensing and metrology \citep{pezze2018}, designing pulses for variational quantum algorithms \citep{magann2021}, and realizing fast qubit reset and readout \citep{gautier2025}. Because hardware parameters, available controls, and cost landscapes vary substantially across platforms, pulse design is typically cast as a numerical optimization problem \citep{koch2022,glaser2015} and solved by methods such as GRAPE \citep{khaneja2005}, Krotov's method \citep{reich2012}, adjoint-state techniques \citep{gautier2025}, or reinforcement learning \citep{ernst2025}. Improvements in the quality of these solutions translate directly into deeper achievable circuits and better-performing quantum devices.

In any realistic setting, a central challenge facing quantum optimal 
control is decoherence: the irreversible coupling of the system to its 
environment, which degrades coherence on timescales comparable to the 
pulse duration and ultimately limits the fidelity of any protocol \citep{krantz2019, PhysRevLett.129.150504}. The evolution of a system subject to environmental noise can be modeled by a stochastic Schr\"odinger equation, a stochastic differential 
equation whose noise terms capture the system's coupling to each 
decoherence channel. Different control protocols 
can route the system through regions of state space that differ 
substantially in their exposure to decoherence \citep{lidar1998,viola1999}, 
so the trajectory through Hilbert space, not just the target state, is an object of optimization. Favorable routes pass through \emph{decoherence-free subspaces} (DFSs) \citep{lidar1998}, on which the noise acts trivially. When a noise model has such structure, or an approximate version of it, a controller can exploit it, and this is the setting in which the regularizers proposed here are effective. When no such structure is present, the gains are small or absent (Sections~\ref{sec:zero_sets} and~\ref{sec:exp_chain}).

We propose QMaxCal, a framework for path-space regularization that
acts directly on the trajectory distribution of an open quantum
system. The SSE generates, for each decoherence channel, a classical
measurement record whose drift encodes the noise experienced by the
system; the records of two evolutions sharing the same decoherence
channels are It\^o diffusions with shared noise structure and 
control-dependent drifts, so Girsanov's theorem
\citep{girsanov1960} yields a closed-form, differentiable estimator
of the KL divergence between their path distributions\rev{ (Figure~\ref{fig:method})}. The framework
can be read as a quantum-control instance of the Maximum Caliber
principle \citep{jaynes1980minimum, jaynes1985macroscopic, presse2013}\rev{.} The two resulting
regularizers contrast with standard penalties on pulse \emph{fluence}
\citep{khaneja2005,reich2012,villanueva2024} and evolution
\emph{smoothness} \citep{ernst2025}: they penalize the observable
consequences of control on the decoherence channels rather than the
control amplitude or its derivatives.

Girsanov's theorem underlies path-integral control \citep{kappen2005,kappen2016adaptive}; to our knowledge its only quantum-control application is Quantum Diffusion Control (QDC) \citep{villanueva2024}, which requires an invertible gauge transformation mapping the Lindblad operators onto the control Hamiltonians.

We make the following contributions:
\begin{enumerate}
    \item We derive a \textbf{closed-form estimator for the KL 
    divergence over quantum trajectory space} via Girsanov's theorem 
    (Eq.~\ref{eq:kl_quantum}), valid for any pair of evolutions 
    sharing the same Lindblad operators\rev{, which recovers the QDC fluence penalty as a special case (Appendix~\ref{app:reduction})}.
    \item We instantiate the general estimator with two physically 
    motivated reference measures, yielding two regularizers that act 
    on different structural features of the noise: the 
    \textbf{Wiener KL regularizer} $(\klw)$ and the 
    \textbf{drift-variance regularizer} $(\dvr)$. We prove that, for dephasing and decay channels, $\klw$ vanishes exactly on the joint kernel of the Lindblad operators and $\dvr$ exactly on decoherence-free evolutions (Theorem~\ref{thm:zero_sets}).
    \item We \textbf{demonstrate the framework on five open-system 
    control benchmarks}, including amplitude damping, STIRAP, a four-level shelving structure, a synthetic qubit chain with site-dependent noise, 
    and a six-qubit chain calibrated to a published snapshot of the 
    IBM Kingston processor, where the regularizers reduce infidelity relative to unregularized gradient-based optimization by up to 50\% (four-level diamond system), by up to 26\% on a four-qubit chain, and by $\sim$16\% on the calibrated IBM Kingston chain, and outperform a constrained-RL baseline.
\end{enumerate}

\section{Related Work}
\label{sec:related}

\paragraph{Gradient-based quantum optimal control.} GRAPE \citep{khaneja2005} and its open-system variant Open GRAPE \citep{boutin2017} compute optimal pulses via piecewise evolution and gradient ascent over density matrices. Krotov's method \citep{reich2012} ensures monotonic convergence. \citet{gautier2025} introduced a scalable adjoint-state method for large open systems. \citet{abdelhafez2019} introduced trajectory-based gradient computation using automatic differentiation through the SSE, demonstrating the quadratic scaling advantage of state vectors over density matrices and enabling GPU parallelization. Our work builds on this computational paradigm but introduces a qualitatively different objective: rather than differentiating a fidelity cost with ad hoc control penalties, we regularize with a path-space KL divergence minimizing the effects of decoherence.

\paragraph{Reinforcement learning.} \citet{ernst2025} cast pulse design as a constrained RL problem and outperform previous gradient-based methods. Their approach introduces penalties on the fluence and smoothness of the control pulse, as well as on the smoothness of the resulting quantum evolution, with the goal of producing experimentally realizable signals and improving computational efficiency. QMaxCal complements this line of work by introducing a regularizer that acts directly on the trajectory through Hilbert space, penalizing the controlled system's coupling to the environment rather than the shape of the drive waveform; the two penalties are not interchangeable, as we verify by direct ablation in Appendix~\ref{app:smoothness_ablation}.

\paragraph{Maximum Caliber and path integral control.} The principle of 
Maximum Caliber (MaxCal) \citep{jaynes1980minimum, jaynes1985macroscopic,presse2013} extends Maximum 
Entropy to dynamical settings, selecting a path distribution that 
minimizes a KL divergence to a reference subject to imposed 
constraints. MaxCal has found applications in classical 
nonequilibrium statistical mechanics \citep{ghosh2020} and molecular 
dynamics \citep{dixit2018,bolhuis2021,brotzakis2021,bolhuis2023}. 
Applied to control with the \emph{uncontrolled} evolution as 
reference, the MaxCal paradigm reduces to path-integral control 
\citep{kappen2005,kappen2016adaptive}, which \citet{villanueva2024} 
adapted to open quantum systems as the Quantum Diffusion Control 
(QDC) algorithm and \citet{najafi2025} recently extended to quantum 
chemistry. Under additional structural conditions on the noise model, 
the resulting KL reduces to a fluence penalty on the controls (Appendix~\ref{app:reduction}), and \citet{villanueva2024} use QDC primarily as an annealing framework for closed-system control. QDC's cost is defined only when an invertible gauge transformation maps the Lindblad operators onto the control Hamiltonians (Appendix~\ref{app:reduction}). None of our benchmarks admits such a transformation, so QDC cannot be run on them; on QDC's own benchmark, our estimator and QDC converge to the same solution (Appendix~\ref{app:qdc_validation}). 
QMaxCal complements this approach by taking the Wiener measure as 
reference and by applying to arbitrary 
Lindblad systems (Section~\ref{sec:wiener}).

\section{Background}
\label{sec:background}

\subsection{Open quantum systems and quantum trajectories}
Open quantum systems are described by \emph{density matrices} $\rho \in \mathbb{C}^{N \times N}$. When the system interacts with a 
Markovian (memoryless) environment, $\rho$ evolves according to the 
Lindblad master equation \citep{lindblad1976} (see Appendix~\ref{app:quantum_primer} for 
a self-contained introduction)
\begin{equation}
\label{eq:lindblad}
\dot{\rho} = -i[H^{(\theta)}(t),\rho] + \sum_{k=1}^K \mathcal{D}_k(\rho), \qquad 
\mathcal{D}_k(\rho) = L_k \rho L_k^\dagger - 
\tfrac{1}{2}\{L_k^\dagger L_k, \rho\},
\end{equation}
where $H^{(\theta)}(t) = H_0 + \sum_a u_a^{(\theta)}(t) H_a$ is the system 
Hamiltonian with time-dependent control fields $u_a^{(\theta)}(t)$ 
depending on parameters $\theta$, and $\{L_k\}_{k=1}^K$ are Lindblad 
operators which model the irreversible interactions with 
the environment (e.g., energy decay, phase randomization).

The master equation describes the ensemble-averaged evolution, but it 
can be \emph{unravelled} into quantum trajectories describing the evolution of individual \emph{pure states} $\ket{\psi(t)} \in \mathbb{C}^N = \mathcal{H}$ under a given interaction with the environment. A given 
Lindblad evolution admits multiple unravellings; throughout this work 
we use the diffusive (homodyne) unravelling, under which a state obeys the stochastic Schr\"odinger 
equation (SSE) \citep{wiseman2009,percival1998}
\begin{align}
\label{eq:sse}
d\ket{\psi(t)} &= \underbrace{-iH^{(\theta)}(t)\ket{\psi(t)}\,\dt - 
\frac{1}{2}\sum_{k=1}^K \left(L_k^\dagger L_k - \alpha_k(t) 
L_k + \tfrac{1}{4}\alpha_k(t)^2\right)\ket{\psi(t)}\,\dt
}_{\text{deterministic drift}} \nonumber\\
&\quad + \underbrace{\sum_k \big(L_k - \tfrac{1}{2}\alpha_k(t)\big)\ket{\psi(t)}\,
\dW_k(t)}_{\text{stochastic diffusion}},
\end{align}
where $\alpha_k(t) = \bra{\psi(t)}(L_k + L_k^\dagger)\ket{\psi(t)}$ and $\dW_k(t)$ 
are independent Wiener increments satisfying $\E[\dW_k(t)] = 0$, 
$\E[\dW_j(t) \dW_k(t)] = \delta_{jk}\,\dt$. The SSE induces a trajectory distribution $P_\theta$ on path space; averaging over $P_\theta$ recovers the density matrix evolution: 
$\rho(t) = \E_{P_\theta}\!\left[\ket{\psi(t)}\bra{\psi(t)}\right]$ 
satisfies the Lindblad master equation~\eqref{eq:lindblad}. The trajectory representation also has a numerical advantage:
state vectors carry $O(\dim)$ rather than $O(\dim^2)$ entries, so a forward SSE pass reaches roughly two qubits beyond the largest system that the Lindblad reference solver fits on a single GPU when the noise levels are not too high (12 versus 10 qubits; training needs more memory, Appendix~\rev{\ref{app:scaling:train}}).

In experimental platforms where the environment is itself measured 
(circuit QED, trapped ions, optical cavities), each individual 
trajectory $\ket{\psi(t)}$ corresponds to a specific realization of 
the classical \emph{measurement record} $I_k(t)$ obtained by 
monitoring decoherence channel $k$. The same Wiener increments 
$\dW_k$ that drive the SSE~\eqref{eq:sse} also drive the measurement record 
\citep{wiseman2009}:
\begin{equation}
\label{eq:record}
dI_k(t) = \alpha_k(t)\,\dt + \dW_k,
\qquad \alpha_k(t) = \bra{\psi(t)}(L_k + L_k^\dagger)\ket{\psi(t)},
\end{equation}
with state-dependent drift $\alpha_k(t)$ and unit diffusion 
coefficient that is independent of the control protocol.

Equation~\eqref{eq:record} is the central object of this paper: 
different controls produce different drifts $\alpha_k$ while sharing 
the same noise structure. This is the setting in which Girsanov's theorem yields a tractable Radon--Nikodym derivative between path distributions, which we state next.
\subsection{Girsanov's theorem}

We state Girsanov's theorem informally here; a formal derivation with all conditions verified for the quantum setting is given in Appendix~\ref{app:girsanov_formal}. Consider two It\^o diffusion processes sharing the same noise but with different drift processes:
\begin{equation}
dX_t^{(i)} = a^{(i)}_t\,\dt + \dW_t, \qquad i = 1,2,
\end{equation}
where each drift $a^{(i)}_t$ may depend on the entire trajectory history up to time $t$. Girsanov's theorem \citep{girsanov1960,oksendal2003} states that the Radon--Nikodym derivative between the induced path measures $P^{(1)}, P^{(2)}$ is
\begin{equation}
\label{eq:girsanov_classical}
\log\frac{dP^{(1)}}{dP^{(2)}} = \int_0^T (a^{(1)}_t - a^{(2)}_t)\,dB^{(2)}_t - \frac{1}{2}\int_0^T\left(a^{(1)}_t - a^{(2)}_t\right)^2\dt,
\end{equation}
\rev{where $B^{(2)}_t := X_t - \int_0^t a^{(2)}_s\,ds$ is a Brownian motion under $P^{(2)}$,} and the KL divergence takes the closed form
\begin{equation}
\label{eq:kl_general}
\KL[P^{(1)} \| P^{(2)}] = \frac{1}{2}\E_{P^{(1)}}\left[\int_0^T \left(a^{(1)}_t - a^{(2)}_t\right)^2\,\dt\right].
\end{equation}

\section{Methods}
\label{sec:methods}
\subsection{Girsanov KL divergence over quantum trajectory space}

Applying Girsanov's theorem (Section~\ref{sec:background}, Eq.~\eqref{eq:kl_general}) to the 
measurement records of two evolutions sharing the same Lindblad 
operators gives our estimator of the KL between two path distributions. Let $P^{(1)}$ and 
$P^{(2)}$ denote the path measures induced 
by the SSE~\eqref{eq:sse} under two control protocols, two Hamiltonians, 
or a controlled-vs.-uncontrolled pair; the only requirement is that the 
Lindblad operators $\{L_k\}$ are shared, which is automatically satisfied 
for any pair of control protocols on the same physical system. Under 
each evolution $i \in \{1, 2\}$, the measurement record for channel $k$ 
has drift $\alpha_k^{(i)}(t) = \bra{\psi^{(i)}(t)}(L_k + L_k^\dagger)
\ket{\psi^{(i)}(t)}$ and unit diffusion coefficient. Girsanov's theorem 
then yields:
\begin{equation}
\label{eq:kl_quantum}
\boxed{\KL[P^{(1)} \| P^{(2)}] = \frac{1}{2}\sum_{k=1}^K \E_{P^{(1)}}\!\left[
\int_0^T \big(\alpha_k^{(1)}(t) - \alpha_k^{(2)}(t)\big)^2\,\dt
\right],}
\end{equation}
where the subscript denotes an average with respect to $P^{(1)}$.

\subsection{The Wiener KL and drift-variance regularizers}
\label{sec:wiener}
Consider the trajectory distribution $P_{\theta}$ introduced above in Section~\ref{sec:background}. The general estimator~\eqref{eq:kl_quantum} requires a choice of 
reference measure. The simplest is the \emph{Wiener measure} 
$P_W$ (the law of pure Brownian motion with zero drift), which yields
\begin{equation}
\label{eq:kl_wiener}
\klw:=\KL[P_\theta \| P_W] = \frac{1}{2}\sum_{k=1}^K \E_{P_\theta}\!\left[
\int_0^T \big(\alpha_k^{(\theta)}(t)\big)^2\,\dt\right],
\end{equation}
and requires no reference SSE simulation. A more permissive reference 
is the closest constant-drift process: for $c \in \R^K$, let $P_c$ 
denote the law of $dY_t = c\,\dt + \dW_t$. The KL, $\KL[P_\theta \| 
P_c]$, is quadratic in $c$, and its minimum over $c$ defines the 
\emph{drift-variance regularizer}:
\begin{equation}
\label{eq:kl_dv}
\dvr := \min_{c \in \R^K} \KL[P_\theta \| P_c] = 
\frac{1}{2}\sum_{k=1}^K \E_{P_\theta}\!\left[\int_0^T \big(
\alpha_k^{(\theta)}(t) - \bar\alpha_k\big)^2\,\dt\right],
\end{equation}
with optimum at $c_k^\star = \bar\alpha_k := T^{-1}\E_{P_\theta}[
\int_0^T \alpha_k^{(\theta)}(t)\,\dt]$, the time-and-ensemble drift 
mean (Appendix~\ref{app:drift_variance}).

\rev{
\subsection{Decoherence-free evolutions and the zero sets of the regularizers}
\label{sec:zero_sets}
\label{sec:dfs}
For certain noise models there exist subspaces on which the Lindblad operators act trivially. A subspace $\mathcal{S} \subseteq \mathcal{H}$ is a \emph{decoherence-free subspace} (DFS) \citep{lidar1998} if each Lindblad operator acts as a scalar on it: $L_k\ket{\phi} = \ell_k\ket{\phi}$ for all $\ket{\phi} \in \mathcal{S}$ and all $k$, with constants $\ell_k \in \mathbb{C}$ depending only on the channel. On such a subspace the dissipator causes no decoherence, and the stochastic backaction $(L_k - \tfrac{1}{2}\alpha_k)\ket{\psi}\,\dW_k$ in the SSE~\eqref{eq:sse} reduces to a global stochastic phase $i\,\mathrm{Im}(\ell_k)\ket{\psi}\,\dW_k$ acting identically on every state in $\mathcal{S}$, so it is physically unobservable\footnote{Two state vectors that differ only by a global phase describe the same physical state.}; it is identically zero when $\ell_k \in \mathbb{R}$, as in all benchmarks of this paper. We call a trajectory a \emph{DFS evolution} if $L_k\ket{\psi(t)} = \ell_k\ket{\psi(t)}$ for all $k$ and all $t \in [0,T]$, with constants $\ell_k$ that do not depend on $t$ or on the noise realization, i.e.\ if it remains in a DFS. It then undergoes unitary decoherence-free evolution with constant drift $\alpha_k = 2\,\mathrm{Re}(\ell_k)$, and if all $\ell_k = 0$ it lies in the joint kernel $\bigcap_k \ker(L_k)$\footnote{The kernel of an operator $L$ is the subspace 
$\ker(L) = \{\ket{\psi} : L\ket{\psi} = 0\}$; the joint kernel 
$\bigcap_k \ker(L_k)$ consists of states annihilated by all 
Lindblad operators simultaneously.}.

The following theorem states when the two regularizers vanish.

\begin{theorem}[Zero sets]
\label{thm:zero_sets}
Let $P_\theta$ be the path measure induced by the SSE~\eqref{eq:sse} under the controls $u^{(\theta)}$, $P_W$ the Wiener measure, and $\klw$ and $\dvr$ the regularizers~\eqref{eq:kl_wiener} and~\eqref{eq:kl_dv}, with record drifts $\alpha_k = \bra{\psi}X_k\ket{\psi}$, $X_k = L_k + L_k^\dagger$. Then the following hold.
\\[2pt]
(a) \emph{Joint kernel $\Rightarrow$ $\klw = 0$.} If $\ket{\psi(t)} \in \bigcap_k \ker(L_k)$ for all $t$, $P_\theta$-almost surely, then $\alpha_k \equiv 0$ for all $k$ and $\klw = 0$.\\[2pt]
(b) \emph{DFS evolution $\Rightarrow$ $\dvr = 0$.} If $P_\theta$-almost every trajectory is a DFS evolution, then $\dvr = 0$.\\[2pt]
(c) \emph{Converses for dephasing and decay.} Suppose that every Lindblad operator is either Hermitian up to a positive factor, $L_k = \sqrt{\kappa_k}\,A_k$ with $A_k = A_k^\dagger$ (dephasing), or nilpotent, $L_k^2 = 0$ (decay, e.g.\ $\sigma_-$ or $\ket{g}\bra{e}$), and that no monitored quadrature vanishes, $X_k \neq 0$ for all $k$. Then (i) if $\dvr = 0$, then $P_\theta$-almost every trajectory is a DFS evolution, with $\ell_k = 0$ for every nilpotent channel; and (ii) if $\klw = 0$, then $\ket{\psi(t)} \in \bigcap_k \ker(L_k)$ for all $t$, $P_\theta$-almost surely.
\end{theorem}

Parts (a) and (b) hold for arbitrary channels. Every benchmark in this paper satisfies the hypothesis of (c)\rev{ (decay channels, plus dephasing on the qubit chains)}. For these channels, $\klw$ therefore vanishes exactly on the joint kernel and $\dvr$ exactly on DFS evolutions; under heterodyne detection this holds for arbitrary Lindblad operators (Corollary~\ref{cor:heterodyne}). The proof is given in Appendix~\ref{app:zero_set_proof}.

\paragraph{Local-oscillator phase.} Homodyne detection has a free local-oscillator (LO) phase $\phi_k$ per channel, which replaces $L_k$ by $e^{-i\phi_k}L_k$ in the unravelling. Theorem~\ref{thm:zero_sets} applies at every LO phase, with $X_k$ replaced by $X_k^\phi = e^{-i\phi_k}L_k + e^{i\phi_k}L_k^\dagger$ (for a Hermitian channel this requires $\cos\phi_k \neq 0$). The zero sets of the regularizers are therefore invariant under the LO phase, but the penalty values away from them, and hence the full optimization, depend on it (Appendix~\ref{app:lo_phase}).

\paragraph{Choosing a regularizer.} Theorem~\ref{thm:zero_sets} indicates which regularizer to use. $\klw$ is the appropriate choice when a joint kernel exists and is reachable by control, such as a ground or dark state next to lossy levels (Sections~\ref{sec:exp_ampdamp} and~\ref{sec:exp_dark}). $\dvr$ vanishes on every DFS evolution, including those with nonzero drift, and is the appropriate choice when no joint kernel exists, for example under strong rate asymmetry or collective noise (Section~\ref{sec:exp_chain}); Appendix~\ref{app:checklist} lists these noise-model patterns. Without a joint kernel, $\klw$ can be harmful: it cannot vanish (part (c)(ii)), and minimizing it pulls the state toward low-drift states, which need not be protected. For pure dephasing, $L = \sqrt{\kappa}\,\sigma_z$, the kernel is trivial and the drift $\alpha = 2\sqrt{\kappa}\expect{\sigma_z}$ vanishes on the equator of the Bloch sphere, where dephasing is strongest, so $\klw$ pushes the state toward the worst region.

}

\subsection{The QMaxCal objective}
We instantiate the framework on \emph{state transfer} tasks: 
starting from an initial state $\ket{\psi(0)}$, the 
goal is to drive the system to a target state 
$\ket{\phi_{\mathrm{target}}}$ at a fixed final time $T$. The 
quantity of interest is the final-state fidelity 
$F(\psi(T)) = |\braket{\phi_{\mathrm{target}}}{\psi(T)}|^2$, and 
combining either regularizer of Section~\ref{sec:wiener} with this 
fidelity objective yields the QMaxCal loss
\begin{equation}
\label{eq:objective}
\mathcal{L}(\theta) = \big(1 - \E_{\psi \sim P_\theta}[F(\psi(T))]
\big) + \lambda_{\klw}\,\klw + \lambda_{\dvr}\,\dvr + 
\lambda_{\mathrm{flu}}\,\mathrm{Flu},
\end{equation}
where $\mathrm{Flu} = \sum_a \int_0^T |u_a^{(\theta)}(t)|^2\,
dt$ is a standard control-fluence penalty 
\citep{khaneja2005,reich2012,villanueva2024}, included to bound the 
control amplitude and produce realistic control policies. The 
hyperparameters $\lambda_\klw, \lambda_\dvr \geq 0$ control the 
trade-off between path-space regularization and fidelity; setting both to zero recovers unconstrained gradient-based optimal control \citep{abdelhafez2019}. \rev{In every experiment only one of $\lambda_\klw$ and $\lambda_\dvr$ is nonzero (Appendix~\ref{app:fluence_ablation} removes the fluence term).} 
The parameters $\theta$ may correspond to piecewise-constant pulse 
amplitudes, neural-network weights, or basis-function coefficients; 
optimization proceeds by gradient descent on $\theta$ with all terms 
evaluated on the same sampled SSE trajectory ensemble\rev{, integrated by Euler--Maruyama or an exponential split-step scheme (Appendix~\ref{app:scaling})}.

\section{Experiments}
\label{sec:experiments}

We evaluate QMaxCal on five open-system control benchmarks across three structural families, 
comparing against an unregularized gradient baseline (the GRAPE-style 
trajectory optimizer of \citet{abdelhafez2019}) and the PPO method of 
\citet{ernst2025}, with PPO hyperparameters tuned per-benchmark via a 
two-phase sweep (Appendix~\ref{app:ppo}). Fidelity differences are 
reported in percentage points (pp), $F_A - F_B$; relative infidelity 
reductions as percentages of $1 - F_B$. Figure~\ref{fig:system_schematics} (Appendix~\ref{app:quantum_primer}) shows the level structure of all five benchmarks for reference.
\subsection{Single-qubit amplitude damping}
\label{sec:exp_ampdamp}
Amplitude damping models a single qubit losing energy to its 
environment\rev{,} the dominant noise process in superconducting, 
trapped-ion, and atomic platforms, where it sets the $T_1$ 
relaxation time \citep{krantz2019,kubica2023}. Population flows 
irreversibly from the excited state $\ket{1}$ to the ground state 
$\ket{0}$ at rate $\gamma$, via the Lindblad operator 
$L = \sqrt{\gamma}\,\sigma_-$ where $\sigma_- = \ket{0}\bra{1}$ 
is the qubit lowering operator. Since 
$\sigma_-\ket{0} = 0$, the kernel $\ker L = \mathrm{span}\{\ket{0}\}$ 
is one-dimensional and consists of the ground state alone 
(Figure~\ref{fig:system_schematics}a)\rev{. The record drift is $\alpha = \sqrt{\gamma}\expect{\sigma_x}$, which vanishes on $\ker L$ but also on the rest of the $y$--$z$ great circle of the Bloch sphere; only trajectories that stay in $\ker L$ have $\klw = 0$ (Theorem~\ref{thm:zero_sets}c).} Regularized policies mirror the strategies of 
active qubit reset and just-in-time state preparation used in real 
hardware \citep{magnard2018,egger2018}, where ancillas are 
deliberately held at $\ket{0}$ until needed to limit $T_1$ 
exposure.

We consider the task $\ket{+} \to \ket{Y}$\rev{, i.e.\ preparing} the superposition $\ket{Y} = (\ket{0} + i\ket{1})/\sqrt{2}$ from 
$\ket{+} = (\ket{0} + \ket{1})/\sqrt{2}$ under two control fields 
that drive rotations of the qubit (Pauli operators $\sigma_x$ and 
$\sigma_y$, defined in Appendix~\ref{app:quantum_primer}). Both 
endpoints sit on the Bloch-sphere equator with equal $\ket{0}$ and 
$\ket{1}$ populations, so a noise-naive geodesic stays on the 
equator and is fully exposed to decoherence; a decoherence-aware policy \rev{can reduce its exposure by detouring} through $\ket{0} \in \ker L$. We sweep the dimensionless decoherence 
strength $\gamma T \in \{0.1, 0.5, 1, 2, 5\}$ at fixed $T=1$ with 
$n_{\mathrm{trajs}}=256$ stochastic trajectories per gradient step, 
training each policy for 5000 steps. We compare the unregularized baseline 
($\lambda_\klw=\lambda_\dvr=0$), Wiener KL ($\lambda_\klw=5$), 
drift-variance ($\lambda_\dvr=5$), and PPO on 
fidelity (Appendix~\ref{app:ampdamp}, Table~\ref{tab:ampdamp_fidelity_full}) 
and SSE-trajectory population variance 
(Appendix~\ref{app:ampdamp}, Table~\ref{tab:ampdamp_variance_full})\rev{, with fidelity computed by the exact Lindblad solver and variance from SSE samples}.
Figure~\ref{fig:ampdamp_var} \rev{(Appendix~\ref{app:ampdamp})} confirms the predicted behavior. At 
$\gamma T = 2$, the Wiener KL policy routes the state through 
$\ket{0} \in \ker L$ where the SSE drift and noise both vanish, and 
the time-integrated population variance 
$\sum_k\int_0^T \mathrm{Var}[P_k(t)]\,dt$ drops $15\times$, from 
$0.0321\pm0.0024$ (baseline) to $0.0021\pm0.0000$ (Wiener KL). Because 
fidelity is an average over trajectories and the optimum is a delta 
on $\ket{\phi_{\mathrm{target}}}$ with zero variance, this contraction \rev{contributes to the} fidelity gain. The Wiener KL gap over baseline is largest in the intermediate 
window $\gamma T \in [0.5, 2]$: at lower noise decoherence barely 
affects the final state, and at higher noise \rev{dissipation alone drives every trajectory to} $\ket{0}$. Drift-variance trails the baseline by $0.25$ to $4$\,pp fidelity: 
Figure~\ref{fig:ampdamp_alpha2} shows that the Wiener KL integrand 
$\mathbb{E}_\psi[\alpha(t)^2]$ remains near the baseline level, 
confirming that drift-variance stabilizes $\alpha$ to a roughly 
constant non-zero value rather than driving the trajectory into $\ker L$\rev{; since $\sigma_-$ is nilpotent, $\dvr$ also vanishes only on $\ker L$ (Theorem~\ref{thm:zero_sets}c), so this is a local optimum of the penalty rather than a zero}.
PPO trails Wiener KL at every $\gamma T$, with the gap widening from \rev{$0.1$\,pp} at low noise to \rev{$2.2$\,pp} at $\gamma T = 5$\rev{; it matches the baseline for $\gamma T \leq 1$ and falls behind for $\gamma T \geq 2$ ($1.8$\,pp at $\gamma T = 5$). All results use LO phase $\phi = 0$; other phases are in Appendix~\ref{app:lo_phase}}.

\subsection{Lossy state avoidance: STIRAP and the diamond system}
\label{sec:exp_dark}

Both systems in this section share a structural pattern: the
optimizer must move population between two states that are not
directly coupled, and the available routes pass through one or more
lossy intermediate states. \rev{The long-lived states lie in the joint kernel, where the Wiener KL vanishes, while trajectories that linger in lossy states acquire nonzero drift through the stochastic backaction; empirically, the Wiener KL steers the optimizer toward routes that minimize their occupation.}

\paragraph{STIRAP.} 
Stimulated Raman adiabatic passage (STIRAP) is widely used for 
transferring population between two long-lived states across atomic, 
molecular, and solid-state platforms \citep{vitanov2017,bergmann2019}. 
The standard model is a three-level $\Lambda$ system with 
metastable ground states $\ket{g_1},\ket{g_2}$ coupled through a 
lossy excited state $\ket{e}$ that decays back to $\ket{g_1}$ at 
rate $\gamma$ via $L = \sqrt{\gamma}\ket{g_1}\bra{e}$ 
(Figure~\ref{fig:system_schematics}b). The two ground states are not directly coupled, so any 
protocol must route population through the $\ket{e}$ coupling; 
at finite protocol speed some $\ket{e}$ population is 
unavoidable, and a good protocol minimizes it. \rev{Here $\alpha = 2\sqrt{\gamma}\,\mathrm{Re}(\overline{c_{g_1}}\,c_e)$, with $c_x = \braket{x}{\psi}$, vanishes on $\ket{g_1}$ and $\ket{g_2}$ and becomes nonzero once population in $\ket{e}$ builds coherence with $\ket{g_1}$, so the Wiener KL discourages occupation of the lossy level without the reward shaping used by conventional methods} \citep{abdelhafez2019, ernst2025}\rev{.}

At both noise levels we tested, the baseline and Wiener KL reach 
statistically indistinguishable fidelity. The regularizer's effect 
shows up instead in $\ket{e}$ occupation, which governs photon 
scattering and protocol viability outside the strict adiabatic 
limit \citep{vitanov2017,bergmann2019}. At $\gamma T = 10$, Wiener 
KL reduces peak $\ket{e}$ population by $56\%$ (from $0.097$ to 
$0.043$) and time-integrated $\ket{e}$ exposure by $35\%$ (from 
$0.026$ to $0.017$). At $\gamma T = 1$ both metrics are already low 
under the baseline and the regularizer's effect is correspondingly 
smaller (Appendix~\ref{app:stirap}, Table~\rev{\ref{tab:stirap_multiseed}}).

\paragraph{Diamond system.}
A complementary structural pattern arises when the computational 
states themselves leak slowly but a nearby \rev{long-lived auxiliary state is available}. This setup is common across qubit architectures: 
metastable trapped-ion qubits \citep{allcock2021,yang2022}, 
erasure-conversion qubits \citep{wu2022,levine2024}, and DFS gate designs \citep{ivanov2010} all route control through 
stable subspaces to avoid lossy ones. We model it as a four-level 
system with two \emph{lossy} computational states $\ket{b}, \ket{t}$ 
(the source and target of the protocol), a \emph{safe} auxiliary 
state $\ket{d}$, and a \emph{dump} state $\ket{\mathrm{dump}}$ that 
collects population leaked from $\ket{b}$ or $\ket{t}$ 
(Figure~\ref{fig:system_schematics}c). The leak 
acts identically on both computational states, with rate $\gamma$:
$L_b = \sqrt{\gamma}\,\ket{\mathrm{dump}}\bra{b}$ and 
$L_t = \sqrt{\gamma}\,\ket{\mathrm{dump}}\bra{t}$. Three controls 
couple adjacent levels: $\Omega_{bt}$ on 
$\ket{b}\!\leftrightarrow\!\ket{t}$ (the direct but lossy route), 
and $\Omega_{bd}, \Omega_{dt}$ that together provide an indirect 
route via $\ket{d}$. The task is $\ket{b} \to \ket{t}$.

The unregularized baseline gets trapped in a local minimum that 
routes population through the lossy direct coupling, reaching only 
$F = 0.665$ with $\sim 30\%$ accumulating in $\ket{\mathrm{dump}}$ 
(Figure~\ref{fig:dark_route_compare}, left). \rev{With the Wiener KL, the optimizer finds} the route through $\ket{d}$, improving fidelity 
to $F = 0.834$ ($+17$\,pp, Appendix~\ref{app:diamond}, 
Table~\ref{tab:dr3_multiseed}; Figure~\ref{fig:dark_route_compare}, 
middle). PPO trails at $F = 0.605$ \rev{($0.53$ without its fluence penalty, Appendix~\ref{app:ppo})}.

Wiener KL policies are also \rev{more} robust to noise-model mismatch\rev{. A} policy trained at $\gamma_{\mathrm{train}} = 2$ and evaluated at $\gamma_{\mathrm{test}} = 5$ \rev{($2.5\times$ the training noise)} reaches $F = 0.665$, matching 
the \emph{baseline's training-noise fidelity} ($+27$\,pp over the 
baseline at the same test noise). The advantage grows monotonically 
with $\gamma_{\mathrm{test}}$\rev{, from $+10$ to $+27$\,pp,} as the baseline and PPO \rev{collapse} 
toward $F\!\sim\!0.3$ 
(Figure~\ref{fig:dark_route_compare}, right; full statistics in 
Appendix~\ref{app:diamond}, Table~\ref{tab:dr3_robustness}).

\begin{figure}[t]
    \centering
    \includegraphics[width=\linewidth]{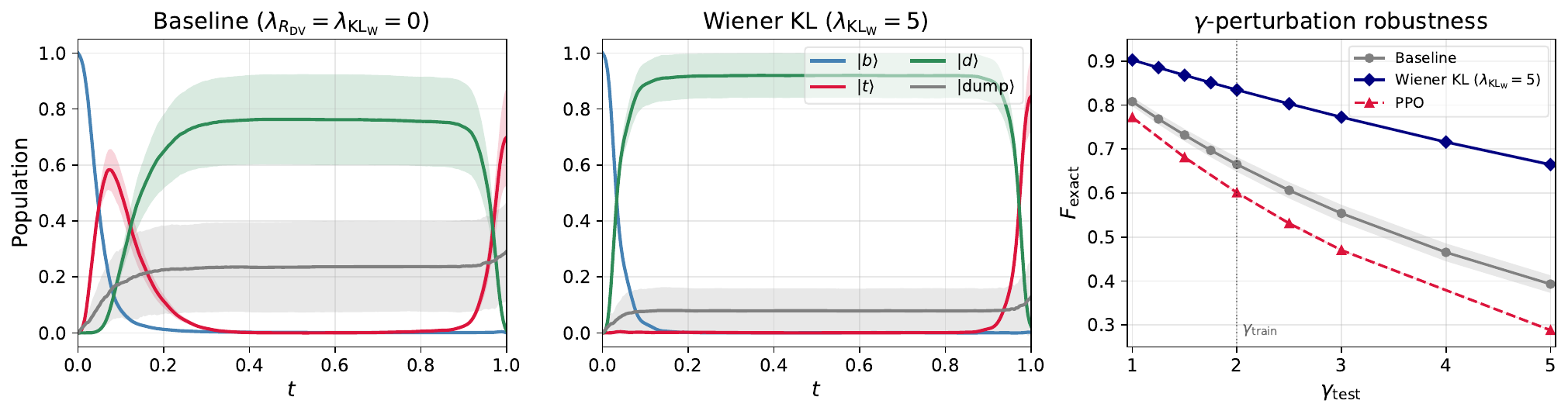}
    \caption{\textbf{Diamond system, $\gamma_{\mathrm{train}}=2$, $T=1$.}
    \emph{Left, middle:} SSE-trajectory populations (mean across 128
    samples, $\pm 1\sigma$) for baseline ($F=0.665$) and Wiener KL
    ($\lambda_{\klw}=5$, $F=0.834$); the regularizer routes population
    through $\ket{d}\in\bigcap_k\ker(L_k)$ rather than the lossy direct coupling.
    \emph{Right:} fidelity vs.\ $\gamma_{\mathrm{test}}$ for policies
    trained at $\gamma=2$; full statistics in
    Appendix~\ref{app:diamond}.}
    \label{fig:dark_route_compare}
\end{figure}

\subsection{Drift-variance on superconducting qubit chains}
\label{sec:exp_chain}
Excitation transfer along a chain of nearest-neighbor qubits is a 
canonical primitive for short-range quantum communication 
\citep{bose2003,lorenzo2015,lyakhov2005}. The task is to move a 
single-qubit excitation from one end of an $N$-qubit chain to the 
other, using $N-1$ tunable nearest-neighbor couplings as controls: 
when active, control $H_a$ swaps the excitation between qubits 
$a-1$ and $a$, leaving the rest of the chain untouched. Each site 
is subject to two local noise channels\rev{,} dephasing (which destroys 
$\ket{0}$/$\ket{1}$ superpositions without changing populations, at rate $\kappa_i$) and 
amplitude damping (the $T_1$ process from \rev{Section~}\ref{sec:exp_ampdamp}, rate $\gamma_i$)\rev{,} with site-specific rates that vary substantially 
across real hardware \citep{krantz2019}. An effective protocol must 
therefore route the excitation around the noisier sites rather than treating the chain uniformly. \rev{Since $\sigma_z$ has no kernel, the Lindblad operators have no joint kernel for the Wiener KL to steer toward (Theorem~\ref{thm:zero_sets}c), so we use the drift-variance regularizer (Section~\ref{sec:zero_sets}).} Lindblad operators and parameter 
values are given in Appendices~\ref{app:chain4_full} and 
\ref{app:ibm_chain}.

\paragraph{Asymmetry sweep.}
On a four-qubit chain ($\dim=16$, eight simultaneous Lindblad 
channels; Figure~\ref{fig:system_schematics}d) we sweep the 
edge-to-interior dephasing ratio 
$\rho = \kappa_{\mathrm{edge}}/\kappa_{\mathrm{int}} \in \{1, 1.5,2, 4, 8\}$, 
where $\kappa_{\mathrm{edge}} = \kappa_0 = \kappa_3$ are the dephasing 
rates on the chain endpoints and $\kappa_{\mathrm{int}} = \kappa_1 = 
\kappa_2$ those on the interior sites. Because of the high 
noise levels in this regime, training is unstable, so we use a 
conservative $\lambda_{\dvr}=0.02$. All training results are reported in Appendix~\ref{app:chain4_full}. 
Asymmetry opens a gap in favor of drift-variance: \rev{at $\rho \in \{1, 1.5\}$ the two policies reach the same fidelity (single seed), and at $\rho = 2$ they agree within one standard deviation}; at $\rho=4$ drift-variance gains $+4.5$\rev{\,}pp ($0.52\pm0.13$ vs.\ $0.48\pm0.02$\rev{, within the seed spread}); at $\rho=8$ the gap widens to $+20$\rev{\,}pp ($0.42\pm0.10$ vs.\ $0.22\pm0.04$), where 
the baseline collapses. The mechanism is visible in the 
populations (Figure~\ref{fig:chain4_asym_pops}): \rev{at $\rho \geq 4$,} drift-variance \rev{moves the excitation off the noisy source site almost immediately and holds it on a single low-noise interior site for most of the protocol}, while the baseline \rev{leaves it on the source site for longer}.

\begin{figure}[t]
    \centering
    \includegraphics[width=\linewidth]{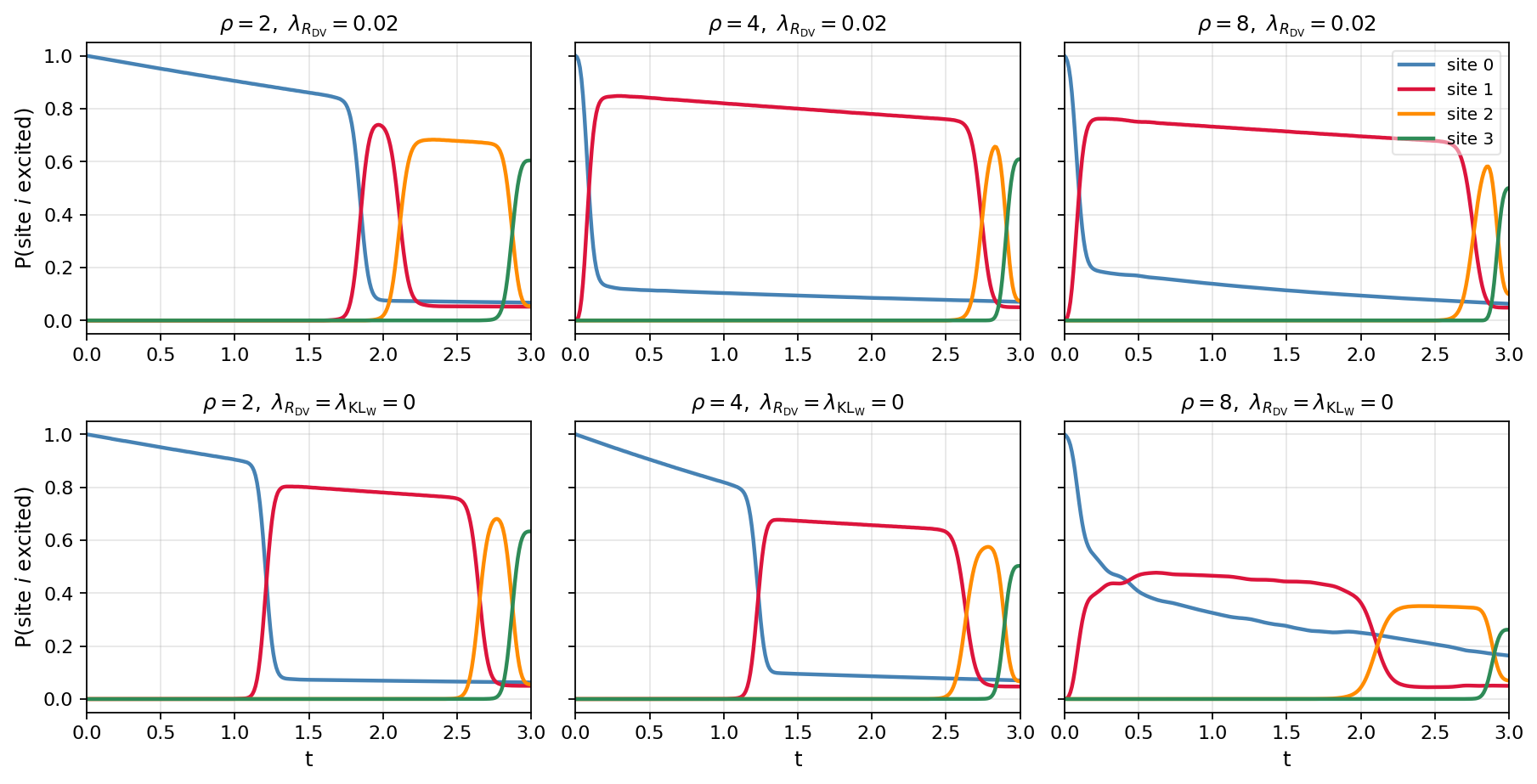}
    \caption{\textbf{Site populations across the asymmetry sweep.}
    Columns: $\rho \in \{2,4,8\}$. Rows: drift-variance
    ($\lambda_\dvr=0.02$, top), baseline (bottom). Each panel shows the single-site excitation population
    $\langle\psi(t)|(I-\sigma_z^{(i)})/2|\psi(t)\rangle$ at sites 
    $i \in \{0,1,2,3\}$ (i.e., the probability that site $i$ holds 
    the excitation)\rev{.} Curves are
    drawn from a representative seed at each $\rho$
    cell (selection rule and per-seed fidelities in
    Appendix~\ref{app:chain4_full}, Table~\ref{tab:chain4_full}).}
    \label{fig:chain4_asym_pops}
\end{figure}

\paragraph{Hardware-calibrated case study.}
We further show that drift-variance regularization improves transfer 
fidelity on a six-qubit chain whose per-site $T_1, T_2$ rates are 
taken from a calibration snapshot of the \texttt{ibm\_kingston} 
processor retrieved on 3 May 2026 from the IBM Quantum 
Platform~\citep{ibmquantum}\rev{ (Appendix~\ref{app:ibm_chain})}. The chain has one site that is 
substantially noisier than the others 
(Figure~\ref{fig:system_schematics}e)\rev{, the kind of strong rate asymmetry for which the asymmetry sweep shows drift-variance helping most}. \rev{The chain has no DFS in the single-excitation manifold in which the transfer takes place; the regularizer instead routes the evolution through an approximate one. We include the Wiener KL as a control: with a trivial joint kernel it cannot vanish (Theorem~\ref{thm:zero_sets}c), and it does not improve on the baseline.} Table~\ref{tab:ibm_chain} 
reports fidelities\rev{, including PPO without its fluence penalty (Appendix~\ref{app:fluence_ablation} removes it for all methods)}; per-site excitation evolutions for all four 
policies are shown in Appendix~\ref{app:ibm_chain}, 
Figure~\ref{fig:chain6_evolutions}.

\begin{table}[h]
\centering
\caption{Six-qubit excitation transfer on IBM Kingston, $T=30\,\mu\mathrm{s}$. Mean fidelity \rev{$\pm$ std} (three seeds\rev{; ten for PPO without fluence penalty}). \rev{All columns except the last use $\lambda_{\mathrm{flu}} = 0.02$.} Bold marks the best in each row.}
\label{tab:ibm_chain}
\small
\setlength{\tabcolsep}{3pt}
\begin{tabular}{lccccc}
\toprule
Transfer & $\lambda_{\dvr}{=}\lambda_{\klw}{=}0$ $F$ & $\lambda_{\dvr}{=}16$ $F$ & $\lambda_{\klw}=0.1$ $F$ & PPO $F$ & \rev{PPO, $\lambda_{\mathrm{flu}}{=}0$ $F$}\\
\midrule
q14$\to$q9  & $0.837 \pm 0.004$ & $\mathbf{0.863 \pm 0.002}$ & $0.831\pm0.004$ & $0.604 \pm 0.041$ & \rev{$0.640 \pm 0.031$}\\
\bottomrule
\end{tabular}
\end{table}

\section{Conclusion}
\label{sec:conclusion}

We derived a closed-form, differentiable estimator for the KL 
divergence between the trajectory distributions of any two open 
quantum system evolutions sharing the same decoherence channels. 
Instantiated with two physically motivated reference measures, this 
estimator yields the QMaxCal framework: a pair of regularizers that 
act on different structural features of the noise. The Wiener KL 
regularizer $\klw$ \rev{vanishes on} the joint kernel $\bigcap_k \ker(L_k)$ and is most effective when this manifold is reachable by control, while the drift-variance regularizer $\dvr$ \rev{vanishes on every DFS evolution and, for dephasing and decay channels, only on those (Theorem~\ref{thm:zero_sets})}. Both are qualitatively distinct from existing penalties 
on control fluence or smoothness: they penalize the observable 
consequences of control on the decoherence channels rather than the 
control amplitude itself. Empirically, the two regularizers improve 
final-state fidelity by up to $+17$\,pp \rev{on the four-level diamond system} 
(rising to $+27$\,pp under a $2.5\times$ noise-model mismatch) and 
$+20$\,pp on multi-channel chains over \rev{the unregularized gradient-based baseline}, with a $\sim$16\% infidelity reduction 
on a six-qubit chain calibrated to a published snapshot of the IBM 
Kingston processor, and discover decoherence-avoiding policies 
without system-specific reward engineering. \rev{Because QMaxCal simulates state vectors rather than density matrices, its forward pass runs on $12$-qubit ($\dim = 4096$) systems on one $40$\,GB A100 GPU, where the density-matrix solver of our PPO baseline runs out of memory (Appendix~\ref{app:scaling}).}

The regularizers help most when the noise model has noise 
structure that the policy can exploit: a reachable joint kernel 
(amplitude damping, diamond), a strongly lossy intermediate state 
(STIRAP), or pronounced site-to-site asymmetry in the decoherence 
rates (four-qubit chain at $\rho \geq 4$, IBM Kingston). When the noise is small or if there is no DFS to exploit (four-qubit chain at 
$\rho \in \{1, 1.5, 2\}$) the gains are small or absent. \rev{This is consistent with Theorem~\ref{thm:zero_sets}.}
\paragraph{Limitations.} \rev{Five} points to note: (i) the Girsanov 
estimator is derived for the diffusive unravelling of the SSE; jump 
unravellings such as photon counting produce piecewise-deterministic 
trajectories, and Eq.~\eqref{eq:kl_quantum} does not apply directly. 
(ii) The Lindblad assumption excludes non-Markovian environments. 
\rev{(iii) The penalty values depend on the LO phase of the homodyne measurement, although, for dephasing and decay channels, their zero sets do not (Appendices~\ref{app:zero_set_proof} and~\ref{app:lo_phase}). (iv) Identifying exploitable structure in advance relies on heuristics (Appendix~\ref{app:checklist}). (v) Gradient-based training must store the forward trajectory, so its memory ceiling lies below the forward-simulation limit (Appendix~\ref{app:scaling:train}).}

\paragraph{Future directions.} \rev{On hardware, where $\ket{\psi(t)}$ is not available, the estimator can be replaced by contrastive density-ratio estimation: a discriminator trained to separate measured records from reference samples gives a variational lower bound on $\KL[P_\theta \| P_W]$ or $\KL[P_\theta \| P_{\bar\alpha}]$ \citep{belghazi2018}. Jump unravellings (marked point processes) and non-Markovian environments (system-plus-environment embeddings) are natural extensions.}

\begin{ack}
The authors thank the anonymous reviewers for their constructive feedback. The research of MM is funded by the Dutch Institute for Emergent Phenomena (DIEP) cluster at the University of Amsterdam via the DIEP programme Foundations and Applications of Emergence (FAEME). The research of MC and ZM is supported by the Vici grant (number VI.C.232.117) from the Dutch Research Council (NWO) and the Academia Sinica Grant for Innovative Applications of AI in Humanities and Scientific Research (AS-IAIA-114-M02). This work used the Dutch national e-infrastructure with the support of the SURF Cooperative using grant no. EINF-14283.
\end{ack}

\bibliographystyle{unsrtnat}
\bibliography{references}

\appendix

\section{Open quantum systems and stochastic unravellings}
\label{app:quantum_primer}

This appendix gives the physical background for the objects used in
the main text. It is self-contained for a reader familiar with
stochastic differential equations but not with quantum mechanics. For
comprehensive treatments, see \citet{wiseman2009} for open quantum
systems and continuous measurement, \citet{percival1998} for quantum
state diffusion, \citet{manzano2022} for a review of
continuous monitoring frameworks, and \citet{koch2022} for quantum
optimal control. Figure~\ref{fig:system_schematics} shows the level
structure and noise topology of all five benchmarks at a glance, for
reference throughout the paper.

\paragraph{Closed quantum systems.} An isolated quantum system of
dimension $N$ is described by a unit vector $\ket{\psi} \in
\mathbb{C}^N$ called the \emph{state vector}, evolving under the
Schr\"odinger equation $\dot{\ket{\psi}} = -iH\ket{\psi}$ for some
Hermitian operator $H$ called the \emph{Hamiltonian}. The evolution is
unitary: norms are preserved and the dynamics are reversible.

\paragraph{Open quantum systems and the density matrix.} In practice,
every quantum system couples to an external environment
(electromagnetic field modes, phonons, nearby spins), and this
coupling introduces irreversible dynamics that cannot be described by
a state vector alone. The standard framework represents the state as
a \emph{density matrix} $\rho \in \mathbb{C}^{N \times N}$: a positive
semidefinite, trace-one, Hermitian matrix. Pure states correspond to
rank-one projectors $\rho = \ket{\psi}\bra{\psi}$; mixed states to
convex combinations $\rho = \sum_i p_i \ket{\psi_i}\bra{\psi_i}$
representing classical uncertainty over which pure state the system
is in.

We illustrate with a single qubit ($N = 2$), with computational basis
$\ket{0} = \begin{psmallmatrix} 1 \\ 0 \end{psmallmatrix}$ and
$\ket{1} = \begin{psmallmatrix} 0 \\ 1 \end{psmallmatrix}$. The
single-qubit operators we use throughout are the Pauli matrices
\begin{equation}
\sigma_x = \begin{pmatrix} 0 & 1 \\ 1 & 0 \end{pmatrix}, \quad
\sigma_y = \begin{pmatrix} 0 & -i \\ i & 0 \end{pmatrix}, \quad
\sigma_z = \begin{pmatrix} 1 & 0 \\ 0 & -1 \end{pmatrix},
\end{equation}
together with the raising and lowering operators
$\sigma_\pm = \tfrac{1}{2}(\sigma_x \mp i\sigma_y)$, so that
$\sigma_- = \ket{0}\bra{1}$ and $\sigma_+ = \ket{1}\bra{0}$. A qubit
density matrix
\begin{equation}
\rho = \begin{pmatrix} \rho_{00} & \rho_{01} \\
\rho_{10} & \rho_{11} \end{pmatrix}
\end{equation}
has diagonal entries $\rho_{00}, \rho_{11} \geq 0$ giving the
probabilities of measuring $\ket{0}$ or $\ket{1}$ (the
\emph{populations}, with $\rho_{00} + \rho_{11} = 1$), and
off-diagonal entries $\rho_{01} = \overline{\rho_{10}}$ encoding
quantum superposition (the \emph{coherences}).

\paragraph{The Lindblad master equation.} The Lindblad master
equation~\eqref{eq:lindblad} is the most general Markovian,
trace-preserving, completely positive evolution for a density
matrix~\citep{lindblad1976,gorini1976}. It consists of two parts.
The commutator $-i[H, \rho]$ generates unitary evolution under the
Hamiltonian $H = H_0 + \sum_a u_a(t) H_a$, where $H_0$ is the free
Hamiltonian and $u_a(t)$ are time-dependent \emph{control fields}
chosen by the experimenter and coupled through fixed control
Hamiltonians $H_a$. The dissipator $\sum_k \mathcal{D}_k(\rho)$ with
$\mathcal{D}_k(\rho) = L_k \rho L_k^\dagger - \frac{1}{2}\{L_k^\dagger
L_k, \rho\}$ encodes the irreversible effect of the environment, with
one \emph{Lindblad operator} $L_k$ per \emph{decoherence channel}:
$L_k$ specifies what the environment does to the system. The five 
benchmarks of Section~\ref{sec:experiments} use Lindblad operators of two basic forms\rev{:} depopulating transitions like $\sigma_-$ that 
move probability mass between levels, and dephasing operators like 
$\sigma_z$ that scramble superpositions without changing populations.
Figure~\ref{fig:system_schematics} shows where each appears.

\begin{figure}[t]
\centering
\begin{tikzpicture}[
    level/.style={very thick, draw=black},
    control/.style={<->, thick, blue!70!black},
    decay/.style={->, thick, red!70!black, decorate, decoration={snake, amplitude=0.6mm, segment length=2mm, post length=1.5mm}},
    statelabel/.style={font=\small, anchor=west},
    panellabel/.style={font=\bfseries\small, anchor=north},
    sublabel/.style={font=\scriptsize, anchor=center, text=blue!70!black}
]

\begin{scope}[shift={(0,0)}]
    \draw[level] (0,0) -- (1.4,0);
    \draw[level] (0,1.6) -- (1.4,1.6);
    \node[statelabel] at (1.5,0) {$\ket{0}$};
    \node[statelabel] at (1.5,1.6) {$\ket{1}$};
    \draw[decay] (0.4,1.5) -- (0.4,0.1);
    \node[font=\scriptsize, text=red!70!black, anchor=east] at (0.35,0.8) {$\gamma$};
    \draw[control] (0.95,0.05) -- (0.95,1.55);
    \node[sublabel, anchor=west] at (1.05,0.8) {$\sigma_x, \sigma_y$};
    \node[panellabel] at (1.0,-0.5) {(a) Amplitude damping};
\end{scope}

\begin{scope}[shift={(4.5,0)}]
    \draw[level] (0,0) -- (0.9,0);
    \draw[level] (1.7,0) -- (2.6,0);
    \draw[level] (0.85,1.6) -- (1.75,1.6);
    \node[font=\small, anchor=north] at (0.45,-0.05) {$\ket{g_1}$};
    \node[font=\small, anchor=north] at (2.15,-0.05) {$\ket{g_2}$};
    \node[font=\small, anchor=south] at (1.3,1.65) {$\ket{e}$};
    \draw[control] (0.7,0.05) -- (1.05,1.55);
    \draw[control] (1.55,1.55) -- (1.9,0.05);
    \draw[decay] (1.0,1.5) .. controls (0.3,1.0) and (0.0,0.6) .. (0.2,0.1);
    \node[font=\scriptsize, text=red!70!black, anchor=east] at (-0.05,0.85) {$\gamma$};
    \node[panellabel] at (1.3,-0.5) {(b) STIRAP};
\end{scope}

\begin{scope}[shift={(9,0)}]
    \draw[level] (-0.4,0.8) -- (0.4,0.8);
    \draw[level] (2.0,0.8) -- (2.8,0.8);
    \draw[level] (0.8,1.8) -- (1.6,1.8);
    \draw[level] (0.8,-0.2) -- (1.6,-0.2);
    \node[font=\small, anchor=east] at (-0.45,0.8) {$\ket{b}$};
    \node[font=\small, anchor=west] at (2.85,0.8) {$\ket{t}$};
    \node[font=\small, anchor=south] at (1.2,1.85) {$\ket{d}$};
    \node[font=\small, anchor=north] at (1.2,-0.25) {$\ket{\mathrm{dump}}$};
    \draw[control] (0.45,0.8) -- (1.95,0.8);
    \node[sublabel, anchor=south] at (1.2,0.85) {$\Omega_{bt}$};
    \draw[control] (0.4,0.95) -- (0.85,1.75);
    \draw[control] (1.55,1.75) -- (2.0,0.95);
    \node[sublabel, anchor=east] at (0.55,1.4) {$\Omega_{bd}$};
    \node[sublabel, anchor=west] at (1.85,1.4) {$\Omega_{dt}$};
    \draw[decay] (0.1,0.7) -- (0.85,-0.1);
    \draw[decay] (2.3,0.7) -- (1.55,-0.1);
    \node[font=\scriptsize, text=red!70!black, anchor=east] at (0.35,0.25) {$\gamma$};
    \node[font=\scriptsize, text=red!70!black, anchor=west] at (2.05,0.25) {$\gamma$};
    \node[panellabel] at (1.2,-0.7) {(c) Diamond};
\end{scope}

\begin{scope}[shift={(0,-3.5)}]
    \foreach \i/\x in {0/0, 1/1.2, 2/2.4, 3/3.6} {
        \draw[level, fill=white] (\x,0.6) circle (0.25);
        \node[font=\small] at (\x,0.6) {\i};
    }
    \foreach \i/\xa/\xb in {0/0/1.2, 1/1.2/2.4, 2/2.4/3.6} {
        \pgfmathsetmacro{\xmid}{(\xa+\xb)/2}
        \draw[control] (\xa+0.25,0.6) -- (\xb-0.25,0.6);
        \node[sublabel, anchor=south] at (\xmid,0.7) {$g_{\i,\pgfmathparse{int(\i+1)}\pgfmathresult}$};
    }
    \draw[decay, line width=1.2pt] (0,0.3) -- (0,-0.4);
    \draw[decay] (1.2,0.3) -- (1.2,-0.4);
    \draw[decay] (2.4,0.3) -- (2.4,-0.4);
    \draw[decay, line width=1.2pt] (3.6,0.3) -- (3.6,-0.4);
    \node[font=\scriptsize, text=red!70!black] at (0,-0.6) {$\rho\kappa_{\mathrm{int}}$};
    \node[font=\scriptsize, text=red!70!black] at (1.2,-0.6) {$\kappa_{\mathrm{int}}$};
    \node[font=\scriptsize, text=red!70!black] at (2.4,-0.6) {$\kappa_{\mathrm{int}}$};
    \node[font=\scriptsize, text=red!70!black] at (3.6,-0.6) {$\rho\kappa_{\mathrm{int}}$};
    \node[font=\scriptsize, anchor=south] at (0,1.05) {source};
    \node[font=\scriptsize, anchor=south] at (3.6,1.05) {target};
    \node[panellabel] at (1.8,-1.1) {(d) Chain-4 (asymmetric dephasing)};
\end{scope}

\begin{scope}[shift={(5.5,-3.5)}]
    \foreach \i/\x/\name in {0/0/q14, 1/1.0/q13, 2/2.0/q12, 3/3.0/q11, 4/4.0/q10, 5/5.0/q9} {
        \draw[level, fill=white] (\x,0.6) circle (0.22);
        \node[font=\scriptsize] at (\x,0.6) {\i};
        \node[font=\scriptsize, anchor=south] at (\x,0.85) {\name};
    }
    \foreach \xa/\xb in {0/1.0, 1.0/2.0, 2.0/3.0, 3.0/4.0, 4.0/5.0} {
        \draw[control] (\xa+0.22,0.6) -- (\xb-0.22,0.6);
    }
    \draw[decay] (0,0.35) -- (0,-0.3);
    \draw[decay, line width=0.6pt] (1.0,0.35) -- (1.0,-0.3);
    \draw[decay, line width=0.6pt] (2.0,0.35) -- (2.0,-0.3);
    \draw[decay] (3.0,0.35) -- (3.0,-0.3);
    \draw[decay, line width=2pt] (4.0,0.35) -- (4.0,-0.3);
    \draw[decay] (5.0,0.35) -- (5.0,-0.3);
    \node[font=\scriptsize, text=red!70!black, anchor=north] at (4.0,-0.35) {noisy};
    \node[font=\scriptsize, anchor=south] at (0,1.15) {source};
    \node[font=\scriptsize, anchor=south] at (5.0,1.15) {target};
    \node[panellabel] at (2.5,-1.1) {(e) IBM Kingston chain-6};
\end{scope}

\begin{scope}[shift={(2.5,-5.4)}]
    \draw[control] (0,0) -- (0.7,0);
    \node[font=\scriptsize, anchor=west] at (0.8,0) {control coupling};
    \draw[decay] (3,0) -- (3.7,0);
    \node[font=\scriptsize, anchor=west] at (3.8,0) {decoherence channel};
    \draw[level] (6.4,0) -- (7.1,0);
    \node[font=\scriptsize, anchor=west] at (7.2,0) {energy level};
\end{scope}

\end{tikzpicture}
\caption{\textbf{Level structure and noise structure of the five 
benchmarks.} Blue double arrows mark coherent control couplings; 
red squiggly arrows mark decoherence channels (each acts on the 
density matrix as a Lindblad operator $L_k$). 
\textbf{(a)} Amplitude damping: $\ket{1}$ decays to $\ket{0}$, two 
controls drive arbitrary single-qubit rotations. 
\textbf{(b)} STIRAP: two stable ground states $\ket{g_1}, \ket{g_2}$ 
coupled through a lossy excited state $\ket{e}$ that decays back to 
$\ket{g_1}$. 
\textbf{(c)} Diamond: two lossy computational states $\ket{b}, 
\ket{t}$ each decay to $\ket{\mathrm{dump}}$; the auxiliary state 
$\ket{d}$ is in $\ker L_b \cap \ker L_t$ and provides a safe indirect 
route. 
\textbf{(d)} Chain-4: four qubits with site-specific dephasing; 
edge sites carry rates $\rho\kappa_{\mathrm{int}}$, interior sites 
$\kappa_{\mathrm{int}}$ (amplitude damping not shown for clarity but 
acts on every site). 
\textbf{(e)} IBM Kingston chain-6: six qubits with per-site $T_1, 
T_2$ from the calibration snapshot; site 4 (q10) is markedly noisier 
($T_2 = 17\,\mu\mathrm{s}$, indicated by the thicker decay arrow).}
\label{fig:system_schematics}
\end{figure}

We illustrate with two channels that recur throughout the paper, both
acting on a single qubit (Figure~\ref{fig:system_schematics}a uses 
amplitude damping; the chain panels d, e use both).
\begin{itemize}
\item \emph{Dephasing}, $L = \sqrt{\kappa}\,\sigma_z$, models an
environment that continuously probes whether the qubit is in $\ket{0}$
or $\ket{1}$ without exchanging energy with it. The populations
$\rho_{00}, \rho_{11}$ are unchanged, while the coherences $\rho_{01},
\rho_{10}$ decay exponentially at rate $2\kappa$. It is the dominant
noise channel in many superconducting and spin-qubit platforms and
sets the $T_2$ coherence time.

\item \emph{Amplitude damping}, $L = \sqrt{\gamma}\,\sigma_-$, models
the qubit losing a quantum of energy to its environment, for example
a photon emitted into the electromagnetic field. Population flows
irreversibly from $\ket{1}$ to $\ket{0}$ at rate $\gamma$, and the
ground state is a fixed point of the dissipation. This dominates in
optical and atomic systems and sets the $T_1$ relaxation time in
superconducting qubits.
\end{itemize}

For an $n$-qubit system, $N = 2^n$, and the density matrix has $N^2 =
4^n$ real parameters; propagating the Lindblad equation directly
therefore requires $\mathcal{O}(N^2) = \mathcal{O}(4^n)$ memory and, with dense operators, $\mathcal{O}(N^3)$ operations per time step.

\paragraph{Unravelling the master equation.} The Lindblad equation
describes the ensemble-averaged dynamics. In many experimental
platforms (circuit QED, trapped ions, optical cavities) the
environment is itself \emph{continuously monitored}, and each run of
the experiment produces a specific measurement record and a
corresponding conditional pure state $\ket{\psi(t)} \in \mathbb{C}^N$.
The decomposition of the density-matrix evolution into a distribution
over such pure-state trajectories is called an \emph{unravelling};
averaging over trajectories always recovers the density matrix,
\begin{equation}
\label{eq:rho_from_traj}
\rho(t) = \mathbb{E}\big[\ket{\psi(t)}\bra{\psi(t)}\big],
\end{equation}
regardless of which unravelling is chosen.

The choice of unravelling is determined by the type of measurement
performed on the environment. Two principal classes exist:
\begin{itemize}
\item \emph{Jump unravelling} (photon counting / direct detection).
The environment is monitored by counting detection events. Between
events, $\ket{\psi}$ evolves smoothly under a non-Hermitian effective
Hamiltonian; at each event, the state jumps as $\ket{\psi} \to
L_k\ket{\psi} / \|L_k\ket{\psi}\|$. Trajectories are piecewise
deterministic with jump times governed by Poisson processes.
\item \emph{Diffusive unravelling} (homodyne / heterodyne detection).
The output field is mixed with a strong local oscillator before
detection, producing a continuous photocurrent. The conditional state
evolves under the stochastic Schr\"odinger
equation~\eqref{eq:sse}, driven by independent Wiener increments
$\dW_k$.
\end{itemize}

\paragraph{Why diffusive unravelling.} The argument of
Section~\ref{sec:background} relies on the measurement record being an
It\^o diffusion with state-dependent drift and unit, control-independent
diffusion coefficient (Eq.~\ref{eq:record}); this is the structure
required by Girsanov's theorem. Under the jump unravelling the
measurement record consists of point processes (click times) rather
than diffusions, and the change of measure for marked point processes
takes a different form, involving products of intensity ratios at jump
times rather than It\^o integrals. Extending the framework in this
direction is discussed in Section~\ref{sec:conclusion}.

\paragraph{State vectors versus density matrices.} The trajectory
representation underlies the computational strategy adopted throughout
this paper. The conditional state vector $\ket{\psi(t)} \in
\mathbb{C}^N$ has $N$ complex entries while the density matrix
$\rho(t) \in \mathbb{C}^{N \times N}$ has $N^2$, so the state space
of a single trajectory is quadratically smaller than that of the
master equation. For $n$-qubit systems with $N = 2^n$, this is the
difference between $\mathcal{O}(2^n)$ and $\mathcal{O}(4^n)$ memory,
and the gap matters in particular for autodiff-based control: gradient
computation through the solver requires storing the entire forward
trajectory, and the smaller state space lets a single GPU handle
larger systems than direct integration of Eq.~\eqref{eq:lindblad} would
allow~\citep{abdelhafez2019}. Trajectories are also independent
across the ensemble and parallelize accordingly. We adopt this
framework, augmented with the path-space regularizer derived in the
main text.

\section{Formal derivation of the Girsanov KL on path space}
\label{app:girsanov_formal}

We give a rigorous path-space derivation of Eq.~\eqref{eq:kl_quantum}.
The argument is standard stochastic analysis~\citep{oksendal2003}; the
state-dependent drift requires the path-space framing to be handled
rigorously. We construct the relevant exponential martingale
(Lemma~\ref{lem:Z_martingale}), use it to relate $P^{(1)}$ and
$P^{(2)}$ as a Girsanov change of measure
(Theorem~\ref{thm:girsanov_app}), and read off the KL divergence
(Corollary~\ref{cor:KL_app}).

\paragraph{Setup.} Let $\Omega := C([0,T];\mathbb{R}^K)$ with
canonical filtration $\mathcal{F}_t$, and let $\mathbb{P}_W$ denote
$K$-dimensional Wiener measure on $\Omega$, under which the coordinate
process $W_k(t) := \omega_k(t)$ is a standard Brownian motion. Let $\ket{\psi^{(i)}_t(\omega)}$ denote the solution of the SSE under protocol
$i \in \{1,2\}$ \rev{driven by the record $\omega$, i.e.\ with $\dW_k$ replaced by $d\omega_k - \alpha_k\,\dt$, $\alpha_k = \bra{\psi}(L_k + L_k^\dagger)\ket{\psi}$ (the quantum filter)}, and define the drift functionals
\begin{equation}
A^{(i)}_k(\omega, t) := \bra{\psi^{(i)}_t(\omega)}(L_k + L_k^\dagger)
\ket{\psi^{(i)}_t(\omega)},
\end{equation}
which are $\mathcal{F}_t$-adapted and uniformly bounded by
$C_k := \|L_k + L_k^\dagger\|_{\mathrm{op}}$. Define $P^{(i)}$ as the \rev{law of the record $I^{(i)}$, the solution of $dI^{(i)}_k = A^{(i)}_k(I^{(i)}, t)\,\dt + dW_k$ under $\mathbb{P}_W$}. By construction, under $P^{(i)}$ the coordinate process
$\omega_k(t)$ has the law of $\int_0^t A^{(i)}_k\,ds + W_k(t)$ under
$\mathbb{P}_W$, so setting $B^{(i)}_k(t) := \omega_k(t) - \int_0^t
A^{(i)}_k(\omega, s)\,ds$ gives a standard $K$-dimensional Brownian
motion under $P^{(i)}$ and
\begin{equation}
\label{eq:Pi_sde}
d\omega_k(t) = A^{(i)}_k(\omega, t)\,\dt + dB^{(i)}_k.
\end{equation}
Set $\Delta_k := A^{(1)}_k - A^{(2)}_k$, so $|\Delta_k| \leq 2 C_k$.

\paragraph{The exponential martingale.} With $B^{(2)}$ the standard
Brownian motion under $P^{(2)}$ from~\eqref{eq:Pi_sde}, define
\begin{equation}
\label{eq:Z_def}
Z_t := \exp\!\left(\sum_k\int_0^t \Delta_k\,dB^{(2)}_k
- \frac{1}{2}\sum_k\int_0^t \Delta_k^2\,ds\right).
\end{equation}

\begin{lemma}\label{lem:Z_martingale}
$Z_t$ is a strictly positive uniformly integrable
$(\mathcal{F}_t,P^{(2)})$-martingale with $Z_0 = 1$ and
$\E_{P^{(2)}}[Z_T] = 1$.
\end{lemma}

\begin{proof}
Apply It\^o's formula to $Z_t = e^{M_t}$ with $M_t := \sum_k\int_0^t
\Delta_k\,dB^{(2)}_k - \frac{1}{2}\sum_k\int_0^t \Delta_k^2\,ds$.
Since $d\langle M\rangle_t = \sum_k \Delta_k^2\,dt$ exactly cancels the
drift in $dM_t$,
\begin{equation}
dZ_t = Z_t \sum_k \Delta_k\,dB^{(2)}_k,
\end{equation}
with $Z_0 = 1$. The bound $\|\Delta(t)\|^2 \leq 4\sum_k C_k^2$ gives
the deterministic estimate
\begin{equation}
\E_{P^{(2)}}\!\left[\exp\!\left(\frac{1}{2}\int_0^T \|\Delta\|^2
\,\dt\right)\right] \leq \exp\!\left(2T\sum_k C_k^2\right) < \infty,
\end{equation}
verifying Novikov's condition, so $Z_t$ is a true uniformly integrable
martingale on $[0,T]$~\citep[Thm.~8.6.5]{oksendal2003} with
$\E_{P^{(2)}}[Z_T] = Z_0 = 1$. Strict positivity is immediate from
the exponential form.
\end{proof}

\begin{theorem}[Girsanov transform on path space]
\label{thm:girsanov_app}
$P^{(1)}$ and $P^{(2)}$ are mutually absolutely continuous, with
\begin{equation}
\label{eq:RN_app}
\frac{dP^{(1)}}{dP^{(2)}}(\omega) = Z_T(\omega).
\end{equation}
\end{theorem}

\begin{proof}
By Lemma~\ref{lem:Z_martingale}, $\mathcal{Q} := Z_T\,dP^{(2)}$
defines a probability measure on $(\Omega, \mathcal{F}_T)$.
Girsanov's theorem~\citep[Thm.~8.6.6]{oksendal2003}, applied to the
$P^{(2)}$-Brownian motion $B^{(2)}$ with adapted bounded shift
$\Delta$, gives that
\begin{equation}
B^{(1)}_k(t) := B^{(2)}_k(t) - \int_0^t \Delta_k(\omega,s)\,ds
= \omega_k(t) - \int_0^t A^{(1)}_k(\omega,s)\,ds
\end{equation}
is a standard $K$-dimensional Brownian motion under $\mathcal{Q}$.
Hence under $\mathcal{Q}$ the coordinate process has the law of
$\int_0^\cdot A^{(1)}_k\,ds + B^{(1)}_k$ with $B^{(1)}$ Brownian, which
is precisely the law of the pushforward $\omega \mapsto I^{(1)}(\omega)$
of $\mathbb{P}_W$, i.e.\ $P^{(1)}$. Therefore $\mathcal{Q} = P^{(1)}$,
giving $dP^{(1)}/dP^{(2)} = Z_T$. The reverse direction is immediate
from $Z_T > 0$ $P^{(2)}$-a.s., so $dP^{(2)}/dP^{(1)} = Z_T^{-1}$.
\end{proof}

\begin{corollary}[Path KL divergence]
\label{cor:KL_app}
\begin{equation}
\KL[P^{(1)} \| P^{(2)}] = \frac{1}{2}\sum_k
\E_{P^{(1)}}\!\left[\int_0^T \Delta_k(\omega,t)^2\,\dt\right].
\end{equation}
\end{corollary}

\begin{proof}
By Theorem~\ref{thm:girsanov_app},
\begin{equation}
\begin{split}
\KL[P^{(1)}\|P^{(2)}] &= \E_{P^{(1)}}\!\left[\log\frac{dP^{(1)}}{dP^{(2)}}
\right] = \E_{P^{(1)}}[\log Z_T] \\
&= \E_{P^{(1)}}\!\left[\sum_k\int_0^T \Delta_k\,dB^{(2)}_k
- \frac{1}{2}\int_0^T \|\Delta\|^2\,\dt\right].
\end{split}
\end{equation}
Under $P^{(1)}$, $B^{(1)}_k$ is Brownian, so $dB^{(2)}_k = dB^{(1)}_k
+ \Delta_k\,\dt$. Substituting,
\begin{equation}
\KL[P^{(1)}\|P^{(2)}] = \E_{P^{(1)}}\!\left[\sum_k\int_0^T
\Delta_k\,dB^{(1)}_k\right] + \frac{1}{2}\E_{P^{(1)}}\!\left[\int_0^T
\|\Delta\|^2\,\dt\right].
\end{equation}
Since $|\Delta_k| \leq 2C_k$ is bounded, $\int_0^\cdot \Delta_k\,
dB^{(1)}_k$ is a true $P^{(1)}$-martingale, so the first term vanishes.
\end{proof}

The Wiener KL of Section~\ref{sec:wiener} is the special
case $A^{(2)} \equiv 0$, where the pushforward map is the identity, so
$P^{(2)} = \mathbb{P}_W$ and $\Delta_k = A^{(1)}_k$,
recovering~\eqref{eq:kl_wiener}.

\section{Derivation of the drift-variance regularizer}
\label{app:drift_variance}
We derive the closed form of the drift-variance
regularizer~\eqref{eq:kl_dv} by minimizing $\KL[P_\theta \| P_c]$
over the constant drift $c \in \R^K$, where $P_c$ denotes the law of
$dY_t = c\,\dt + dW_t$ as defined in Section~\ref{sec:wiener}.
Substituting $\alpha_k^{(1)}(t) = \alpha_k^{(\theta)}(t)$ and
$\alpha_k^{(2)}(t) = c_k$ into the multi-channel Girsanov
KL~\eqref{eq:kl_quantum} (suppressing the $(\theta)$ superscript and
$t$ argument hereafter for readability),
\begin{equation}
\label{eq:kl_pc_app}
\KL[P_\theta \| P_c] = \frac{1}{2}\sum_k \E_{P_\theta}\!\left[
\int_0^T (\alpha_k - c_k)^2\,\dt\right].
\end{equation}
The KL decomposes as a sum over channels, each depending only on its
own $c_k$, so we minimize per channel. Expanding the square,
\begin{align}
\label{eq:expand_app}
\E_{P_\theta}\!\left[\int_0^T (\alpha_k - c_k)^2\,\dt\right]
&= \E_{P_\theta}\!\left[\int_0^T \alpha_k^2\,\dt\right]
- 2c_k\,\E_{P_\theta}\!\left[\int_0^T \alpha_k\,\dt\right]
+ c_k^2 T,
\end{align}
which is quadratic in $c_k$ with positive leading coefficient $T$.
Differentiating and setting to zero gives the optimal constant drift
\begin{equation}
\label{eq:cstar_app}
c_k^\star = \frac{1}{T}\,\E_{P_\theta}\!\left[\int_0^T
\alpha_k\,\dt\right] =: \bar\alpha_k,
\end{equation}
the time-and-ensemble drift mean. Substituting~\eqref{eq:cstar_app}
back into~\eqref{eq:kl_pc_app},
\begin{equation}
\dvr := \min_{c \in \R^K} \KL[P_\theta \| P_c]
= \frac{1}{2}\sum_k \E_{P_\theta}\!\left[\int_0^T (\alpha_k -
\bar\alpha_k)^2\,\dt\right],
\end{equation}
which is~\eqref{eq:kl_dv}. The integrand is the squared deviation of
the instantaneous drift from its time-and-ensemble mean:
$\dvr$ vanishes if and only if $\alpha_k^{(\theta)}(t) =
\bar\alpha_k$ for $(P_\theta \otimes \dt)$-almost every $(\omega, t)$
and every channel $k$, i.e.\ the trajectory drift is constant in time
and across realizations.

\rev{
\section{Proof of Theorem~\ref{thm:zero_sets}}
\label{app:zero_set_proof}

\paragraph{Setting and standing assumptions.} Throughout, $\dim\mathcal{H} = N < \infty$, the operators $H_0$, $H_a$ and $L_k$ are fixed, and the controls $u_a^{(\theta)}$ are bounded and measurable on $[0,T]$. Fix LO phases $\{\phi_k\}$ and write
\[
\tilde L_k = e^{-i\phi_k}L_k, \qquad X_k = \tilde L_k + \tilde L_k^\dagger, \qquad \alpha_k(t) = \bra{\psi(t)}X_k\ket{\psi(t)},
\]
so that $\alpha_k$ is the drift~\eqref{eq:lo_drift}. Replacing $L_k$ by $\tilde L_k$ leaves the dissipator in~\eqref{eq:lindblad} unchanged, and the homodyne unravelling at phase $\phi_k$ is the SSE~\eqref{eq:sse} with $L_k$ replaced by $\tilde L_k$. Under $P_\theta$, the innovations $W_1, \dots, W_K$ are independent standard Brownian motions with respect to a filtration $(\mathcal{F}_t)$, and the SSE has a unique strong solution $\ket{\psi(t)}$ with continuous paths and $\|\psi(t)\| = 1$: its coefficients are polynomial in $\psi$, hence locally Lipschitz, so a unique strong solution exists up to a possible explosion time, and since the unit sphere is invariant under the SSE the solution stays bounded and exists on all of $[0,T]$. We write the SSE as
\begin{equation}
\label{eq:sse_generic}
d\ket{\psi} = \mu(t,\psi)\,\dt + \sum_{j=1}^K \sigma_j(t,\psi)\,\dW_j, \qquad \sigma_j = \big(\tilde L_j - \tfrac{1}{2}\alpha_j\big)\ket{\psi},
\end{equation}
as in Eq.~\eqref{eq:sse}.

\begin{lemma}[It\^o's rule for expectation values]
\label{lem:ito_expect}
Let $Y$ be Hermitian and $y(t) = \bra{\psi(t)}Y\ket{\psi(t)}$. Then $y$ is a continuous semimartingale,
\[
y(t) = y(0) + \int_0^t b_Y(s)\,ds + \sum_{j=1}^K \int_0^t \beta^Y_j(s)\,\dW_j(s), \qquad \beta^Y_j = \expect{Y\tilde L_j + \tilde L_j^\dagger Y} - y\,\alpha_j,
\]
where $b_Y$ is bounded and adapted, and each $\beta^Y_j$ is bounded, adapted and has continuous paths.
\end{lemma}
\begin{proof}
Identify $\mathbb{C}^N$ with $\mathbb{R}^{2N}$. The map $f(\psi) = \bra{\psi}Y\ket{\psi}$ is a real quadratic form, so It\^o's formula applies and its second-order term is exact: $df = 2\,\mathrm{Re}\bra{\psi}Y\ket{d\psi} + \bra{d\psi}Y\ket{d\psi}$. Inserting~\eqref{eq:sse_generic} and using $\dW_i\,\dW_j = \delta_{ij}\,\dt$ for real Brownian motions gives $b_Y = 2\,\mathrm{Re}\bra{\psi}Y\ket{\mu} + \sum_j \bra{\sigma_j}Y\ket{\sigma_j}$ and
\[
\beta^Y_j = 2\,\mathrm{Re}\bra{\psi}Y\big(\tilde L_j - \tfrac{1}{2}\alpha_j\big)\ket{\psi} = \bra{\psi}Y\tilde L_j\ket{\psi} + \bra{\psi}\tilde L_j^\dagger Y\ket{\psi} - \alpha_j\,y = \expect{Y\tilde L_j + \tilde L_j^\dagger Y} - y\,\alpha_j .
\]
Both coefficients are continuous functions of $\psi(t)$ on the unit sphere (and bounded measurable in $t$ through the controls), hence bounded; $\beta^Y_j$ does not depend on the controls and inherits the continuity of $t \mapsto \psi(t)$.
\end{proof}

\begin{lemma}[Vanishing drift variance]
\label{lem:const_drift}
If $\dvr = 0$, then almost surely $\alpha_k(t) = \bar\alpha_k$ for all $t \in [0,T]$ and all $k$, where $\bar\alpha_k$ is the deterministic constant of Eq.~\eqref{eq:kl_dv}.
\end{lemma}
\begin{proof}
Each term of $\dvr = \frac{1}{2}\sum_k \E\int_0^T (\alpha_k - \bar\alpha_k)^2\,\dt$ is non-negative, so each vanishes. By Tonelli's theorem, $\int_0^T (\alpha_k(t) - \bar\alpha_k)^2\,\dt = 0$ almost surely, i.e.\ $\alpha_k(t) = \bar\alpha_k$ for Lebesgue-almost every $t$. Since $\alpha_k(t) = \bra{\psi(t)}X_k\ket{\psi(t)}$ is continuous in $t$, equality holds for every $t$.
\end{proof}

\begin{lemma}[A constant semimartingale has no diffusion part]
\label{lem:const_semimart}
Let $y(t) = y(0) + \int_0^t b(s)\,ds + \sum_j \int_0^t \beta_j(s)\,\dW_j(s)$ with $b$, $\beta_j$ bounded and adapted, and $\beta_j$ continuous in $t$. If $t \mapsto y(t)$ is almost surely constant on $[0,T]$, then almost surely $\beta_j(t) = 0$ for all $t \in [0,T]$ and all $j$.
\end{lemma}
\begin{proof}
The process $M_t = \sum_j \int_0^t \beta_j\,\dW_j$ is a continuous martingale, and since $y$ is constant, $M_t = -\int_0^t b(s)\,ds$ has paths of finite variation. A continuous martingale of finite variation is constant \citep[Ch.~IV, Prop.~1.2]{revuz1999}, so $M \equiv 0$ and its quadratic variation vanishes. By independence of the $W_j$, $\langle M\rangle_T = \sum_j \int_0^T \beta_j(t)^2\,\dt = 0$, so each $\beta_j(t) = 0$ for almost every $t$, and for every $t$ by continuity.
\end{proof}

\paragraph{Proof of part (a).} Suppose $L_k\ket{\psi(t)} = 0$ for all $k$ and $t$. Then $\bra{\psi}\tilde L_k\ket{\psi} = 0$ and $\bra{\psi}\tilde L_k^\dagger\ket{\psi} = \overline{\bra{\psi}\tilde L_k\ket{\psi}} = 0$, so $\alpha_k(t) = 0$ for every $t$, $k$ and LO phase, and $\klw = 0$ by~\eqref{eq:kl_wiener}. \hfill$\square$

\paragraph{Proof of part (b).} On a DFS evolution, $\bra{\psi(t)}\tilde L_k\ket{\psi(t)} = e^{-i\phi_k}\ell_k$, so $\alpha_k(t) = 2\,\mathrm{Re}(e^{-i\phi_k}\ell_k)$ for all $t$ and almost every realization. Since the $\ell_k$ are deterministic, $\bar\alpha_k = 2\,\mathrm{Re}(e^{-i\phi_k}\ell_k)$, the integrand of~\eqref{eq:kl_dv} vanishes identically, and $\dvr = 0$. \hfill$\square$

\paragraph{Proof of part (c)(i).} Assume $\dvr = 0$. \rev{The hypothesis $X_k \neq 0$ of part (c) holds automatically for a nilpotent $L_k \neq 0$: if $X_k = 0$ then $L_k = -e^{2i\phi_k}L_k^\dagger$, so $0 = L_k^2 = -e^{2i\phi_k}L_k^\dagger L_k$ and $L_k = 0$. For a Hermitian channel $X_k = 2\sqrt{\kappa_k}\cos\phi_k\,A_k$, so it is equivalent to $\cos\phi_k \neq 0$, the form used in Step~2.}

\emph{Step 1.} By Lemma~\ref{lem:const_drift}, each $\alpha_k$ is almost surely constant on $[0,T]$. By Lemma~\ref{lem:ito_expect} with $Y = X_k$, $\alpha_k$ is a semimartingale whose $\dW_j$-coefficients are $\beta_{kj} = \expect{X_k\tilde L_j + \tilde L_j^\dagger X_k} - \alpha_k\alpha_j$. Lemma~\ref{lem:const_semimart} gives $\beta_{kj}(t) = 0$ for all $t$ and $j$, almost surely; we use only the own-channel coefficient
\begin{equation}
\label{eq:beta_kk}
\beta_{kk} = \expect{X_k\tilde L_k + \tilde L_k^\dagger X_k} - \alpha_k^2 = 0 .
\end{equation}

\emph{Step 2 (Hermitian channels).} Let $L_k = \sqrt{\kappa_k}\,A_k$ with $A_k = A_k^\dagger$ and $\cos\phi_k \neq 0$. Then $\tilde L_k = e^{-i\phi_k}\sqrt{\kappa_k}\,A_k$ and $X_k = 2\sqrt{\kappa_k}\cos\phi_k\,A_k$, so $X_k\tilde L_k + \tilde L_k^\dagger X_k = 2\kappa_k\cos\phi_k\,(e^{-i\phi_k} + e^{i\phi_k})A_k^2 = 4\kappa_k\cos^2\!\phi_k\,A_k^2$ and $\alpha_k = 2\sqrt{\kappa_k}\cos\phi_k\expect{A_k}$. Substituting into~\eqref{eq:beta_kk},
\[
0 = \beta_{kk} = 4\kappa_k\cos^2\!\phi_k\big(\expect{A_k^2} - \expect{A_k}^2\big) = 4\kappa_k\cos^2\!\phi_k\,\big\|(A_k - \expect{A_k})\ket{\psi(t)}\big\|^2,
\]
where the last equality uses $A_k = A_k^\dagger$ and $\|\psi\| = 1$. Hence $A_k\ket{\psi(t)} = a_k\ket{\psi(t)}$ with $a_k = \expect{A_k} = \bar\alpha_k/(2\sqrt{\kappa_k}\cos\phi_k)$, which by Step 1 is the same deterministic constant for all $t$ and almost every realization. Thus $L_k\ket{\psi(t)} = \ell_k\ket{\psi(t)}$ with $\ell_k = \sqrt{\kappa_k}\,a_k \in \mathbb{R}$.

\emph{Step 3 (nilpotent channels).} Let $L_k^2 = 0$, hence also $(L_k^\dagger)^2 = 0$. Then
\[
X_k\tilde L_k + \tilde L_k^\dagger X_k = e^{-2i\phi_k}L_k^2 + 2L_k^\dagger L_k + e^{2i\phi_k}(L_k^\dagger)^2 = 2L_k^\dagger L_k,
\]
and~\eqref{eq:beta_kk} gives $n_k(t) := \bra{\psi(t)}L_k^\dagger L_k\ket{\psi(t)} = \alpha_k(t)^2/2 = \bar\alpha_k^2/2$ for all $t$, almost surely. Apply Lemma~\ref{lem:ito_expect} to $Y = N_k := L_k^\dagger L_k$. Since $N_k\tilde L_k = e^{-i\phi_k}L_k^\dagger L_k^2 = 0$ and $\tilde L_k^\dagger N_k = e^{i\phi_k}(L_k^\dagger)^2L_k = 0$, the $\dW_k$-coefficient of $n_k$ is $\beta^{N_k}_k = -n_k\alpha_k$. As $n_k$ is constant, Lemma~\ref{lem:const_semimart} gives $n_k\alpha_k = \bar\alpha_k^3/2 = 0$. Hence $\bar\alpha_k = 0$ and $n_k(t) = \|L_k\ket{\psi(t)}\|^2 = 0$, i.e.\ $L_k\ket{\psi(t)} = 0$ for all $t$. This holds at every LO phase.

\emph{Conclusion.} Every channel is covered by Step 2 or Step 3. Taking the union of the finitely many exceptional null sets, almost every trajectory satisfies $L_k\ket{\psi(t)} = \ell_k\ket{\psi(t)}$ for all $k$ and all $t \in [0,T]$, with deterministic constants $\ell_k$, and $\ell_k = 0$ for every nilpotent channel. This is a DFS evolution. \hfill$\square$

\paragraph{Proof of part (c)(ii).} Setting $c_k = \bar\alpha_k$ in~\eqref{eq:expand_app} and using $\E_{P_\theta}[\int_0^T \alpha_k\,\dt] = T\bar\alpha_k$ gives
\[
\klw = \dvr + \frac{T}{2}\sum_k \bar\alpha_k^2 ,
\]
a sum of two non-negative terms. Hence $\klw = 0$ implies $\dvr = 0$ and $\bar\alpha_k = 0$ for all $k$. By part (c)(i), almost every trajectory is a DFS evolution. For a nilpotent channel $\ell_k = 0$ by (c)(i); for a Hermitian channel Step~2 gives $\ell_k = \sqrt{\kappa_k}\,a_k = \bar\alpha_k/(2\cos\phi_k) = 0$. Thus $L_k\ket{\psi(t)} = 0$ for all $k$ and $t$, i.e.\ $\ket{\psi(t)} \in \bigcap_k \ker(L_k)$. \hfill$\square$

\paragraph{Remarks.} (1) For a Hermitian channel at $\cos\phi_k = 0$ we have $X_k = 0$: the record is pure noise and $\alpha_k \equiv 0$, so both regularizers vanish for every evolution and neither converse holds; this is why part (c) excludes it. (2) Under homodyne detection the converses can fail for channels that are neither Hermitian nor nilpotent. For $L = \tfrac{c}{2}I + iB$ with $c \in \R$ and $B = B^\dagger$, measured at $\phi = 0$, we have $X = cI$ and $\alpha \equiv c$ for every state, so $\dvr = 0$ for every evolution, whereas a DFS evolution requires an eigenvector of $B$: the decoherence acts entirely in the unmeasured quadrature. Heterodyne detection, which records both quadratures, closes this gap.

\paragraph{Heterodyne detection.} Heterodyne detection of channel $k$ \citep{wiseman2009} is homodyne detection of two channels $L_k/\sqrt{2}$ at LO phases $\phi_k$ and $\phi_k + \pi/2$. The two copies give the same dissipator as $L_k$, so the Lindblad equation is unchanged, and the two drifts
\[
\alpha_{k,1} = \sqrt{2}\,\mathrm{Re}\big(e^{-i\phi_k}\expect{L_k}\big), \qquad \alpha_{k,2} = \sqrt{2}\,\mathrm{Im}\big(e^{-i\phi_k}\expect{L_k}\big)
\]
together record the complex number $\expect{L_k}$. Since this is homodyne detection with $2K$ channels, Eq.~\eqref{eq:kl_wiener} gives
\[
\klw^{\mathrm{het}} = \frac{1}{2}\sum_k\sum_{j=1,2}\E_{P_\theta}\!\left[\int_0^T \alpha_{k,j}^2\,\dt\right] = \sum_k \E_{P_\theta}\!\left[\int_0^T |\expect{L_k}|^2\,\dt\right],
\]
using $\alpha_{k,1}^2 + \alpha_{k,2}^2 = 2|\expect{L_k}|^2$, and likewise Eq.~\eqref{eq:kl_dv} gives
\[
\dvr^{\mathrm{het}} = \sum_k \E_{P_\theta}\!\left[\int_0^T |\expect{L_k} - m_k|^2\,\dt\right],
\]
with $m_k = T^{-1}\E_{P_\theta}[\int_0^T \expect{L_k}\,\dt]$. Neither penalty depends on the LO phases.

\begin{corollary}[Heterodyne detection]
\label{cor:heterodyne}
Under heterodyne detection, for arbitrary Lindblad operators $L_k$: (i) $\dvr^{\mathrm{het}} = 0$ if and only if $P_\theta$-almost every trajectory is a DFS evolution; (ii) $\klw^{\mathrm{het}} = 0$ if and only if $\ket{\psi(t)} \in \bigcap_k \ker(L_k)$ for all $t$, $P_\theta$-almost surely.
\end{corollary}
\begin{proof}
The ``if'' directions are parts (a) and (b) of Theorem~\ref{thm:zero_sets} applied to the $2K$ channels $L_k/\sqrt{2}$, which have the same kernels and DFS evolutions as the $L_k$. For the converse of (i), Step 1 of the proof of part (c)(i), applied to the $2K$ channels, gives vanishing own-channel coefficients for both records of channel $k$. For $M = L_k/\sqrt{2}$ measured at phase $\varphi$, this coefficient is
\[
2\big(\expect{M^\dagger M} - |\expect{M}|^2\big) + 2\,\mathrm{Re}\big[e^{-2i\varphi}\big(\expect{M^2} - \expect{M}^2\big)\big],
\]
and the second term changes sign between $\varphi = \phi_k$ and $\varphi = \phi_k + \pi/2$. Adding the two coefficients gives
\[
0 = 2\big(\expect{L_k^\dagger L_k} - |\expect{L_k}|^2\big) = 2\,\big\|(L_k - \expect{L_k})\ket{\psi(t)}\big\|^2,
\]
so $L_k\ket{\psi(t)} = \expect{L_k}\ket{\psi(t)}$. By Lemma~\ref{lem:const_drift}, $\alpha_{k,1}$ and $\alpha_{k,2}$ are almost surely equal to the deterministic constants $\bar\alpha_{k,1}$ and $\bar\alpha_{k,2}$, so $\ell_k := \expect{L_k} = e^{i\phi_k}(\bar\alpha_{k,1} + i\bar\alpha_{k,2})/\sqrt{2}$ is the same for all $t$ and almost every realization, and the evolution is a DFS evolution. For the converse of (ii), the decomposition in the proof of part (c)(ii), applied to the $2K$ channels, shows that $\klw^{\mathrm{het}} = 0$ implies $\dvr^{\mathrm{het}} = 0$ and $\bar\alpha_{k,1} = \bar\alpha_{k,2} = 0$; hence $\ell_k = 0$, i.e.\ $L_k\ket{\psi(t)} = 0$ for all $k$ and $t$.
\end{proof}
}

\rev{
\section{Practitioner's checklist}
\label{app:checklist}
Exact DFS detection is infeasible for multi-qubit systems, and approximate DFSs, which the regularizers also exploit (Section~\ref{sec:exp_chain}), have no general discovery method. Table~\ref{tab:checklist} lists three patterns that can be recognized directly from the physics and classifies our benchmarks under them.

\begin{table}[h]
\centering
\small
\color{revcolor}
\caption{Practitioner's checklist: noise-model patterns, the regularizer each calls for, and the benchmarks of Section~\ref{sec:experiments} that fall under it.}
\label{tab:checklist}
\begin{tabular}{@{}lp{5.0cm}cp{2.6cm}@{}}
\toprule
Pattern & Signature & Use & Benchmarks \\
\midrule
1.\ Safe/lossy states & Some states are annihilated by every $L_k$ (ground or metastable levels), others decay or leak & $\klw$ & Amp.\ damping, STIRAP, diamond \\
2.\ Rate asymmetry & No usable DFS; decoherence rates differ strongly across sites or channels & $\dvr$ & Four-qubit chain, IBM Kingston \\
3.\ Symmetric DFS & Collective noise; the $L_k$ act as nonzero scalars on a subspace & $\dvr$ & Not tested \\
\bottomrule
\end{tabular}
\end{table}

}

\section{Proof: Girsanov KL reduces to fluence in the QDC-compatible case}
\label{app:reduction}

Following \citet{villanueva2024}, consider a Lindblad equation with
$n_c$ Lindblad operators $\{C_a\}$ and real symmetric noise matrix $D$. The
unravelling admits a gauge freedom: an invertible $A \in \text{GL}(n_c,
\mathbb{C})$ transforms $\tilde{C}_b = C_a A_{ab}$, \rev{$\tilde{D} = A^{-1} D A^{-\dagger}$} while leaving the Lindbladian dissipator invariant; the
diffusive SSE in the transformed basis takes the form
\begin{equation}
d\psi = -iH\psi\,\dt 
- \frac{1}{2}\sum_{a,b}\tilde{D}_{ab}
\left(\tilde{C}_b^\dagger \tilde{C}_a 
- 2\tilde{c}_a \tilde{C}_b 
+ \tilde{c}_a \tilde{c}_b\right)\psi\,\dt 
+ \sum_a(\tilde{C}_a - \tilde{c}_a)\psi\,d\tilde{W}_a,
\label{eq:transformed_sse}
\end{equation}
where $\tilde{c}_a = \psi^\dagger \tilde{C}_a^{(h)} \psi$ and 
$\tilde{C}_a^{(h)} = \tfrac{1}{2}(\tilde{C}_a + \tilde{C}_a^\dagger)$
is the Hermitian part. Suppose there exists a choice of $A$ such that 
$\tilde{D}$ remains real symmetric and $\tilde{C}_a = -iH_a$, 
identifying the transformed Lindblad operators with $-i$ times the 
control Hamiltonians. Anti-Hermiticity of $\tilde{C}_a$ implies 
$\tilde{C}_a^{(h)} = 0$, hence $\tilde{c}_a = 0$ identically: the 
nonlinear backaction terms in~\eqref{eq:transformed_sse} drop out and 
the SSE becomes linear in $\psi$. Combining with the Hamiltonian 
control $H = H_0 + u_a H_a$, the diffusion and Hamiltonian terms 
acting along $H_a$ combine into $-i H_a \psi (u_a\,\dt + d\tilde{W}_a)$, 
so the controls $u_a$ enter as deterministic shifts of the Brownian 
increments \citep[Eq.~20]{villanueva2024}. The corresponding 
measurement record in the transformed basis is
\begin{equation}
dY_a^{(u)} = u_a\,\dt + d\tilde{W}_a,
\end{equation}
with drift $u_a$ deterministic and state-independent.

Applying Corollary~\ref{cor:KL_app} to the controlled record $dY_a^{(u)}$ 
versus the uncontrolled record $dY_a^{(0)} = d\tilde{W}_a$ (drifts 
$u_a$ and $0$, both deterministic), the expectation in the KL 
collapses\rev{; since $\tilde{W}$ has covariance $\tilde{D}$,}
\begin{equation}
\KL[P_u \| P_0] = \rev{\frac{1}{2}\int_0^T u(t)^\top \tilde{D}^{-1} u(t)\,\dt,}
\end{equation}
which is the control fluence\rev{ weighted by $\tilde{D}^{-1}$}. The structural condition 
$\tilde{C}_a = -iH_a$ is restrictive: it requires that an invertible 
$A$ exists mapping the original Lindblad operators $\{C_a\}$ to 
$\{-iH_a\}$, which forces the noise channels (after mixing) to act 
in the same algebraic directions as the controls. The Wiener KL 
of Section~\ref{sec:wiener} is a different quantity, computed in 
the original physical basis where the measurement-record drift is 
$\langle L_k + L_k^\dagger\rangle = 2\langle L_k^{(h)}\rangle$, twice 
the expectation of the Hermitian part. This drift is state-dependent 
and generically unrelated to $u_a$; the two KLs coincide only when 
the QDC conditions hold and the change of basis $A$ brings the 
original $\{C_a\}$ to anti-Hermitian form.

\rev{
\section{Validation against QDC in the QDC-compatible case}
\label{app:qdc_validation}
QDC can only be instantiated when the Lindblad operators can be gauge-transformed onto the control Hamiltonians (Appendix~\ref{app:reduction}). None of our benchmarks satisfies this condition, so we compare the two methods on the benchmark of \citet{villanueva2024}, which does: $\ket{X}\to\ket{Y}$ state preparation with $C_1 = \sigma^+$, $C_2 = \sigma^-$, $D = 0.005$, $T = 1$ and $K = 128$ piecewise-constant controls, using $400$ trajectories per importance-sampling step for QDC, as in \citet{villanueva2024}, and $128$ per gradient step for QMaxCal. In this case the Girsanov KL in the anti-Hermitian transformed basis reduces to the fluence (Appendix~\ref{app:reduction}). Figure~\ref{fig:qdc_comparison} shows that both methods converge to the same fidelity (QMaxCal $0.974 \pm 0.013$, QDC $0.975 \pm 0.012$; mean $\pm$ std over $1000$ evaluation trajectories of a single run), consistent with $F^* = 0.9759 \pm 0.0006$ reported by \citet{villanueva2024} (std over $505$ runs), the same path-integral cost, and the same control pulses, with identical fluence and KL.
\begin{figure}[h]
    \centering
    \includegraphics[width=\textwidth]{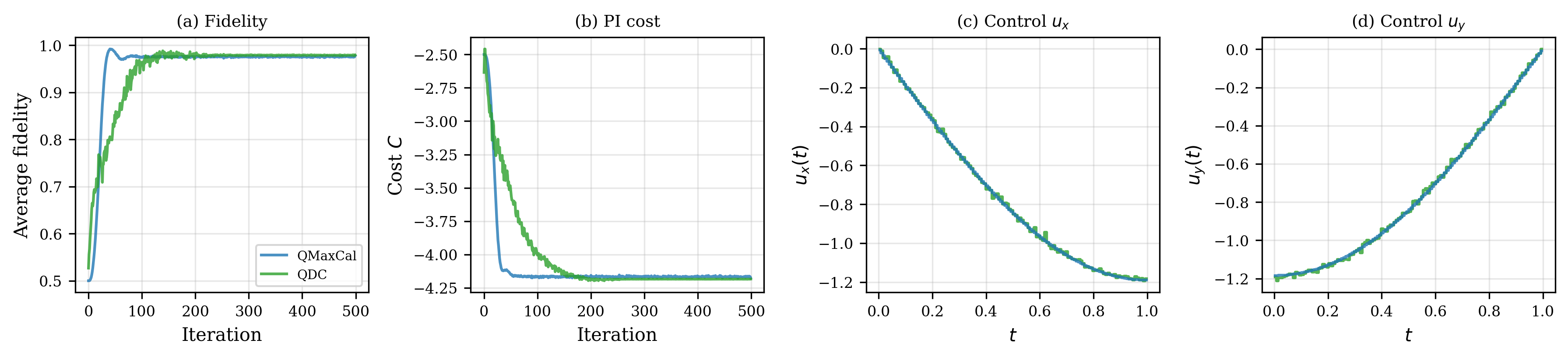}
    \caption{\rev{Comparison of QMaxCal and QDC on the $\ket{X} \to \ket{Y}$ benchmark of \citet{villanueva2024} with QDC-compatible noise ($C_1 = \sigma^+$, $C_2 = \sigma^-$, $D = 0.005$). \textbf{(a)} Average fidelity versus iteration. \textbf{(b)} Path-integral cost $C$ versus iteration. \textbf{(c, d)} Converged control pulses $u_x(t)$ and $u_y(t)$. Both methods converge to the same solution.}}
    \label{fig:qdc_comparison}
\end{figure}

\section{LO-phase dependence of the Wiener KL}
\label{app:lo_phase}
\rev{Homodyne detection of channel~$k$ has a free local-oscillator (LO) phase $\phi_k$, which replaces $L_k$ by $e^{-i\phi_k}L_k$ in the unravelling. The monitored quadrature and record drift become
\begin{equation}
\label{eq:lo_drift}
X_k^{\phi} = e^{-i\phi_k}L_k + e^{i\phi_k}L_k^\dagger, \qquad \alpha_k^{\phi}(t) = \bra{\psi(t)} X_k^{\phi}\ket{\psi(t)}.
\end{equation}
The Lindblad equation~\eqref{eq:lindblad}, its kernels and DFSs, and the type (Hermitian or nilpotent) of each channel do not depend on $\phi$, and Theorem~\ref{thm:zero_sets} holds at any LO phase with $X_k, \alpha_k$ replaced by $X_k^\phi, \alpha_k^\phi$ (Appendix~\ref{app:zero_set_proof}). The condition $X_k^\phi \neq 0$ always holds for nilpotent channels. For a Hermitian channel $X_k^\phi = 2\sqrt{\kappa_k}\cos\phi_k\,A_k$, which vanishes at $\cos\phi_k = 0$: the record is then pure noise and both penalties vanish for every control. The values of the penalties away from their zero sets do depend on $\phi$, but $\klw$ still routes the amplitude-damping trajectory through $\ker L$ at other phases (Table~\ref{tab:lo_phase}).}

The homodyne LO phase $\phi$ changes the monitored quadrature but not the Lindblad equation. For amplitude damping, $X^\phi = e^{-i\phi}\sigma_- + e^{i\phi}\sigma_+ = \cos\phi\,\sigma_x + \sin\phi\,\sigma_y$, so the record drift is $\alpha^\phi = \sqrt{\gamma}\expect{\cos\phi\,\sigma_x + \sin\phi\,\sigma_y}$. Its zero set always contains $\ker L = \mathrm{span}\{\ket{0}\}$ (Theorem~\ref{thm:zero_sets}a), but its shape elsewhere depends on $\phi$; for example, at $\phi = \pm\pi/2$ the magnitude of the drift is maximal at the target $\ket{Y}$. We repeat the amplitude-damping experiment of Section~\ref{sec:exp_ampdamp} with the same hyperparameters at $\gamma = 2$, $T = 1$ for several LO phases, selecting each controller by lowest training loss (Table~\ref{tab:lo_phase}).

At every LO phase, the regularized controller holds the ensemble near $\ker L$ for the middle of the protocol: the mean ground-state population over $t \in [0.25T, 0.75T]$ is $0.98$--$0.99$ at $\lambda_\klw = 1$ and $0.993$ at $\lambda_\klw = 5$, compared with $0.956$ for the unregularized controller. At large $\lambda_\klw$, however, the two nonzero phases over-penalize the route from the kernel to the target and fall below the baseline. At moderate $\lambda_\klw$ the gain from decoherence avoidance is retained at all phases. Final performance therefore depends on the LO phase, but the fidelity gain over the baseline is consistent across phases once $\lambda_\klw$ is moderate.

\begin{table}[h]
\centering
\color{revcolor}
\caption{Amplitude damping, exact-Lindblad fidelity of Wiener-KL policies at three LO phases. Unregularized baseline (single seed): $F = 0.9727$. $\lambda_\klw = 1$: mean $\pm$ std ($n-1$) over three seeds; $\lambda_\klw = 5$: single seed.}
\label{tab:lo_phase}
\begin{tabular}{@{}lccc@{}}
\toprule
 & $\phi = 0$ & $\phi = -\pi/4$ & $\phi = -\pi/2$ \\
\midrule
$\lambda_\klw = 5$ & $0.981$ & $0.956$ & $0.950$ \\
$\lambda_\klw = 1$ & $0.9786 \pm 0.0001$ & $0.9754 \pm 0.0005$ & $0.9779 \pm 0.0004$ \\
\bottomrule
\end{tabular}
\end{table}
}

\section{Additional Experimental Details}
\label{app:datadump}

This appendix collects multi-seed statistics and robustness
experiments referenced in Section~\ref{sec:experiments}, organized
to mirror the body's experiment ordering.
Section~\ref{app:common_setup} fixes the common training setup; the
per-experiment subsections then list only the values that differ.

\subsection{Common training setup}
\label{app:common_setup}

We have the following common architectural choices for all our experiments:

\begin{itemize}
\item \textbf{Controller}: each control channel $u_a^{(\theta)}(t)$ is a
truncated Fourier series in time with $16$ Fourier modes (i.e.\
$16$ cosine and $16$ sine coefficients per channel plus a DC
offset).
\item \textbf{Optimizer}: Adam with learning rate $10^{-3}$.
\item \textbf{Solver}: Euler--Maruyama (EM) or Exponential-Split (ExpSplit), see Appendix~\ref{app:scaling}
\item \textbf{Fidelity reported}: at every snapshot we evaluate the
exact-Lindblad fidelity by integrating the master equation under the
current controller and computing $F = \langle\psi_{\mathrm{tgt}}|\rho(T)|\psi_{\mathrm{tgt}}\rangle$.
\item \textbf{Selection rule}: the reported $F$ for any given run is the maximum over snapshots of the exact-Lindblad fidelity\rev{, except where stated otherwise: the main-text diamond values and the LO-phase runs (Appendix~\ref{app:lo_phase}) use the controller with the lowest training loss, and the fluence sweep (Table~\ref{tab:fluence_sweep}) reports the final step}.
\item \textbf{Hardware}: single NVIDIA A100 GPU (40\,GB);
total project compute reported in Appendix~\ref{app:compute}.
\end{itemize}

Each subsection below lists $T$, the total number of
optimizer steps, the regularizer weights actually swept, the
fluence-warm-up length, and the number of training seeds.

\subsection{Single-qubit amplitude damping}
\label{app:ampdamp}
Tables~\ref{tab:ampdamp_fidelity_full} and~\ref{tab:ampdamp_variance_full} 
report the multi-seed statistics for the fidelity and trajectory-variance 
comparison underlying Section~\ref{sec:exp_ampdamp}.

\begin{figure}[t]
    \centering
    \includegraphics[width=0.9\linewidth]{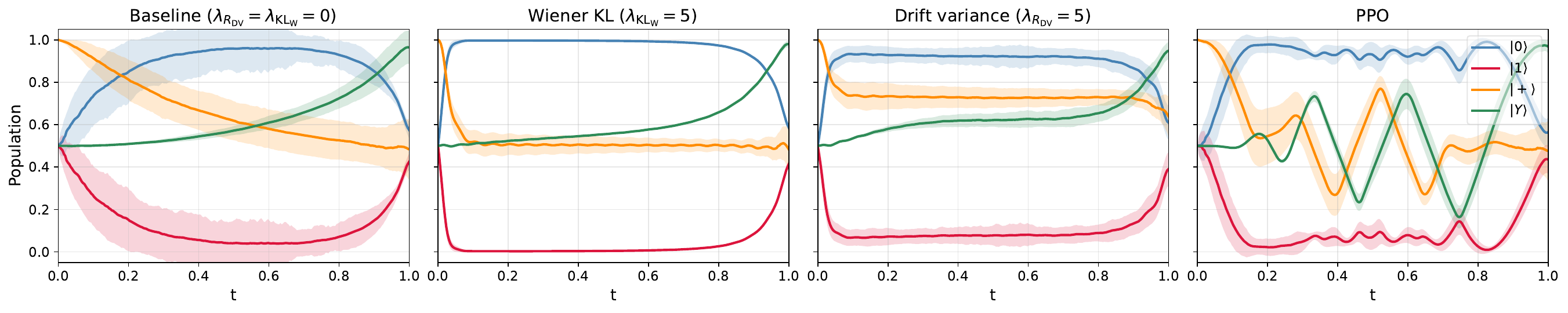}
    \caption{\textbf{SSE-trajectory population variance at 
    $\gamma T = 2$.} Mean populations across $128$ SSE samples (lines) 
    with $\pm 1\sigma$ bands for the baseline ($\lambda_\klw=\lambda_\dvr=0$), Wiener KL 
    ($\lambda_{\klw}=5$), drift-variance ($\lambda_{\dvr}=5$), and PPO. 
    Wiener KL contracts the time-integrated population variance $15\times$\rev{ by routing} the state along $\ker L$\rev{.} Drift-variance reduces the variance 
    relative to baseline but stabilizes outside $\ker L$.}
    \label{fig:ampdamp_var}
\end{figure}

\begin{table}[t]
\centering
\caption{\textbf{Amplitude damping (full statistics):} Fidelity $F$ 
for $\ket{+}\!\to\!\ket{Y}$, mean $\pm$ std across 4 seeds. Bold marks 
the best in each row.}
\label{tab:ampdamp_fidelity_full}
\begin{tabular}{@{}ccccc@{}}
\toprule
$\gamma T$ & $\lambda_{\klw}{=}\lambda_{\dvr}{=}0$ & $\lambda_{\klw}{=}5$ & $\lambda_{\dvr}{=}5$ & PPO \\
\midrule
$0.1$ & $0.9942 \pm 0.0001$ & $\mathbf{0.9952 \pm 0.0001}$ & $0.9917 \pm 0.0001$ & $0.9943 \pm 0.0001$ \\
$0.5$ & $0.9831 \pm 0.0001$ & $\mathbf{0.9900 \pm 0.0001}$ & $0.9698 \pm 0.0001$ & $0.9856 \pm 0.0003$ \\
$1.0$ & $0.9772 \pm 0.0002$ & $\mathbf{0.9863 \pm 0.0001}$ & $0.9555 \pm 0.0001$ & $0.9773 \pm 0.0007$ \\
$2.0$ & $0.9715 \pm 0.0004$ & $\mathbf{0.9810 \pm 0.0001}$ & $0.9458 \pm 0.0004$ & $0.9684 \pm 0.0017$ \\
$5.0$ & $0.9633 \pm 0.0001$ & $\mathbf{0.9673 \pm 0.0001}$ & $0.9222 \pm 0.0027$ & $0.9451 \pm 0.0002$ \\
\bottomrule
\end{tabular}
\end{table}

\begin{table}[t]
\centering
\caption{\textbf{Amplitude damping (full statistics)}. 
SSE-trajectory population variance 
$\sum_k \int_0^T \mathrm{Var}[P_k(t)]\,dt$ for 
$\ket{+}\!\to\!\ket{Y}$, mean $\pm$ std across 4 seeds 
($128$ SSE samples per seed). Bold marks the best in each row.}
\label{tab:ampdamp_variance_full}
\begin{tabular}{@{}ccccc@{}}
\toprule
$\gamma T$ & $\lambda_{\klw}{=}\lambda_{\dvr}{=}0$ & $\lambda_{\klw}{=}5$ & $\lambda_{\dvr}{=}5$ & PPO \\
\midrule
$0.1$ & $0.0032 \pm 0.0000$ & $\mathbf{0.0013 \pm 0.0000}$ & $0.0019 \pm 0.0000$ & $0.0036 \pm 0.0002$ \\
$0.5$ & $0.0106 \pm 0.0001$ & $\mathbf{0.0021 \pm 0.0000}$ & $0.0055 \pm 0.0000$ & $0.0090 \pm 0.0001$ \\
$1.0$ & $0.0166 \pm 0.0002$ & $\mathbf{0.0022 \pm 0.0000}$ & $0.0077 \pm 0.0001$ & $0.0134 \pm 0.0005$ \\
$2.0$ & $0.0321 \pm 0.0024$ & $\mathbf{0.0021 \pm 0.0000}$ & $0.0077 \pm 0.0001$ & $0.0192 \pm 0.0019$ \\
$5.0$ & $0.0411 \pm 0.0061$ & $\mathbf{0.0016 \pm 0.0000}$ & $0.0065 \pm 0.0002$ & $0.0243 \pm 0.0040$ \\
\bottomrule
\end{tabular}
\end{table}

\begin{figure}
    \centering
    \includegraphics[width=\linewidth]{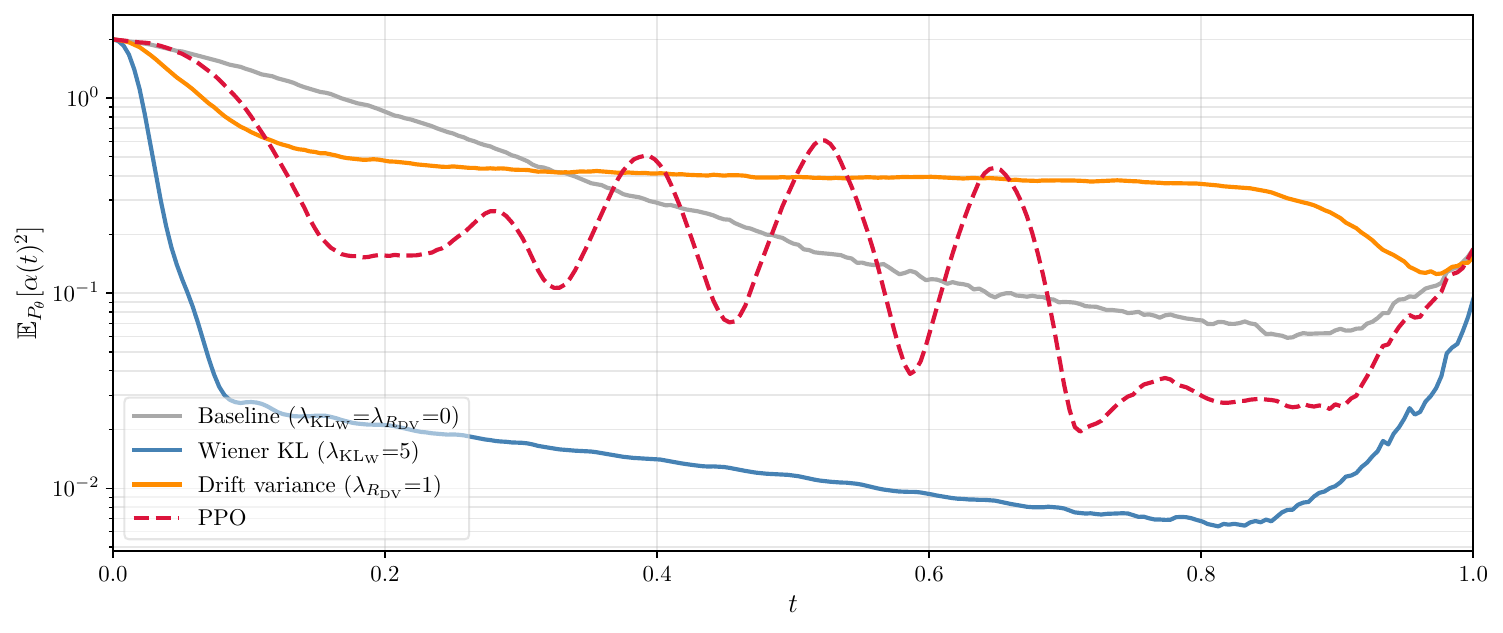}
    \caption{\textbf{Wiener KL integrand $\mathbb{E}_\psi[\alpha(t)^2]$, $\gamma T = 2$.} Trajectory-averaged drift squared on a 
    logarithmic axis. Wiener KL suppresses $\alpha$ toward 
    $\ker L$, reducing the time-integrated value 
    $\tfrac{1}{2}\int\langle\alpha^2\rangle\,dt$ from $0.230$ 
    (baseline) to $0.031$\rev{, a $7\times$ reduction} consistent with the trajectory routing seen in 
    Figure~\ref{fig:ampdamp_var}. Drift-variance remains near 
    the baseline level ($0.237$), while PPO achieves an 
    intermediate reduction ($0.160$).}
    \label{fig:ampdamp_alpha2}
\end{figure}

\paragraph{Experiment-specific parameters.}
\begin{itemize}
\item $T=1.0$ in dimensionless units, so $\gamma T = \gamma$.
\item Decay rate sweep: $\gamma T \in \{0.1, 0.5, 1, 2, 5\}$.
\item Total optimizer steps: $5000$.
\item Solver: ExpSplit
\item Discretized time grid: $256$ points per trajectory, snapshot every
$500$ steps; $256$ SSE samples per gradient step.
\item Regularizer weights: $\lambda_{\klw}\in\{0,5\}$,
$\lambda_{\dvr}\in\{0,1,5\}$, $\lambda_{\mathrm{flu}}=0.01$.
\item Fourier controller: $n_{\mathrm{modes}}=20$, $\mathrm{init\_scale}=0.1$.
\item Seeds: $4$ per cell.
\end{itemize}

\subsection{STIRAP}
\label{app:stirap}

Table~\ref{tab:stirap_multiseed} reports multi-seed fidelity and 
$\ket{e}$-exposure statistics for the two-noise-level comparison in 
Section~\ref{sec:exp_dark}, and Table~\ref{tab:stirap_robustness} 
reports the $\gamma$-perturbation robustness experiment, in which 
protocols trained at $(\gamma T)_{\mathrm{train}} = 10$ are evaluated 
at five test noise levels via exact Lindblad integration.

\begin{table}[t]
\centering
\caption{\textbf{STIRAP results.} Fidelity $F$, peak $\ket{e}$ 
population, and time-integrated $\ket{e}$ exposure for the transfer 
$\ket{g_1}\to\ket{g_2}$, mean $\pm$ std across 3 seeds. Bold marks 
the best in each column per $\gamma T$ block; baseline and Wiener KL 
are statistically indistinguishable in fidelity at both noise levels.}
\label{tab:stirap_multiseed}
\begin{tabular}{@{}lcccc@{}}
\toprule
$\gamma T$ & method & $F_{\mathrm{exact}}$ & peak $\ket{e}$ & $\int_0^T \ket{e}\,dt$ \\
\midrule
$1.0$ & baseline ($\lambda_{\klw}{=}\lambda_{\dvr}{=}0$)  & $0.9959 \pm 0.0003$ & $0.056 \pm 0.016$ & $0.024 \pm 0.000$ \\
$1.0$ & Wiener KL ($\lambda_{\klw}{=}1$) & $\mathbf{0.9961 \pm 0.0005}$ & $\mathbf{0.054 \pm 0.011}$ & $\mathbf{0.022 \pm 0.001}$ \\
$1.0$ & PPO                       & $0.9614 \pm 0.0001$ & $0.160 \pm 0.003$ & $0.164 \pm 0.001$ \\
\midrule
$10.0$ & baseline ($\lambda_{\klw}{=}\lambda_{\dvr}{=}0$)  & $\mathbf{0.9797 \pm 0.0010}$ & $0.097 \pm 0.028$ & $0.026 \pm 0.001$ \\
$10.0$ & Wiener KL ($\lambda_{\klw}{=}1$) & $0.9790 \pm 0.0005$ & $\mathbf{0.043 \pm 0.012}$ & $\mathbf{0.017 \pm 0.001}$ \\
$10.0$ & PPO                       & $0.7289 \pm 0.0000$ & $0.163 \pm 0.010$ & $0.162 \pm 0.001$ \\
\bottomrule
\end{tabular}
\end{table}

\begin{table}[t]
\centering
\caption{\textbf{STIRAP $\gamma$-perturbation robustness with forbidden-state exposure.}
Each protocol trained at $(\gamma T)_{\mathrm{train}} = 10$ and 
evaluated at the indicated test value. Mean $\pm$ std across 3 seeds.}
\label{tab:stirap_robustness}
\begin{tabular}{@{}llccc@{}}
\toprule
$(\gamma T)_{\mathrm{test}}$ & method & $F_{\mathrm{exact}}$ & peak $\ket{e}$ & $\int_0^T \ket{e}\,dt$ \\
\midrule
$2.0$  & baseline ($\lambda_{\klw}{=}\lambda_{\dvr}{=}0$)  & $0.9945 \pm 0.0006$ & $0.113 \pm 0.028$ & $0.023 \pm 0.001$ \\
$2.0$  & Wiener KL ($\lambda_{\klw}{=}1$) & $0.9945 \pm 0.0010$ & $\mathbf{0.042 \pm 0.007}$ & $\mathbf{0.015 \pm 0.001}$ \\
$2.0$  & PPO                       & $0.733 \pm 0.001$   & $0.282 \pm 0.011$ & $0.215 \pm 0.002$ \\
\midrule
$4.0$  & baseline ($\lambda_{\klw}{=}\lambda_{\dvr}{=}0$)  & $0.9904 \pm 0.0006$ & $0.108 \pm 0.028$ & $0.024 \pm 0.001$ \\
$4.0$  & Wiener KL ($\lambda_{\klw}{=}1$) & $0.9906 \pm 0.0006$ & $\mathbf{0.042 \pm 0.007}$ & $\mathbf{0.016 \pm 0.001}$ \\
$4.0$  & PPO                       & $0.784 \pm 0.001$   & $0.239 \pm 0.012$ & $0.199 \pm 0.001$ \\
\midrule
$10.0$ (train) & baseline ($\lambda_{\klw}{=}\lambda_{\dvr}{=}0$)  & $0.9797 \pm 0.0009$ & $0.097 \pm 0.023$ & $0.026 \pm 0.001$ \\
$10.0$ (train) & Wiener KL ($\lambda_{\klw}{=}1$) & $0.9792 \pm 0.0004$ & $\mathbf{0.043 \pm 0.009}$ & $\mathbf{0.017 \pm 0.000}$ \\
$10.0$ (train) & PPO                       & $0.729 \pm 0.000$   & $0.163 \pm 0.010$ & $0.162 \pm 0.001$ \\
\midrule
$20.0$ & baseline ($\lambda_{\klw}{=}\lambda_{\dvr}{=}0$)  & $0.9626 \pm 0.0017$ & $0.084 \pm 0.016$ & $0.026 \pm 0.001$ \\
$20.0$ & Wiener KL ($\lambda_{\klw}{=}1$) & $0.9603 \pm 0.0008$ & $\mathbf{0.047 \pm 0.011}$ & $\mathbf{0.018 \pm 0.000}$ \\
$20.0$ & PPO                       & $0.521 \pm 0.000$   & $0.095 \pm 0.004$ & $0.110 \pm 0.001$ \\
\midrule
$30.0$ & baseline ($\lambda_{\klw}{=}\lambda_{\dvr}{=}0$)  & $0.9446 \pm 0.0028$ & $0.071 \pm 0.013$ & $0.026 \pm 0.001$ \\
$30.0$ & Wiener KL ($\lambda_{\klw}{=}1$) & $0.9410 \pm 0.0013$ & $\mathbf{0.050 \pm 0.014}$ & $\mathbf{0.019 \pm 0.000}$ \\
$30.0$ & PPO                       & $0.369 \pm 0.000$   & $0.061 \pm 0.002$ & $0.075 \pm 0.000$ \\
\bottomrule
\end{tabular}
\end{table}

\begin{figure}[t]
  \centering
  \includegraphics[width=\linewidth]{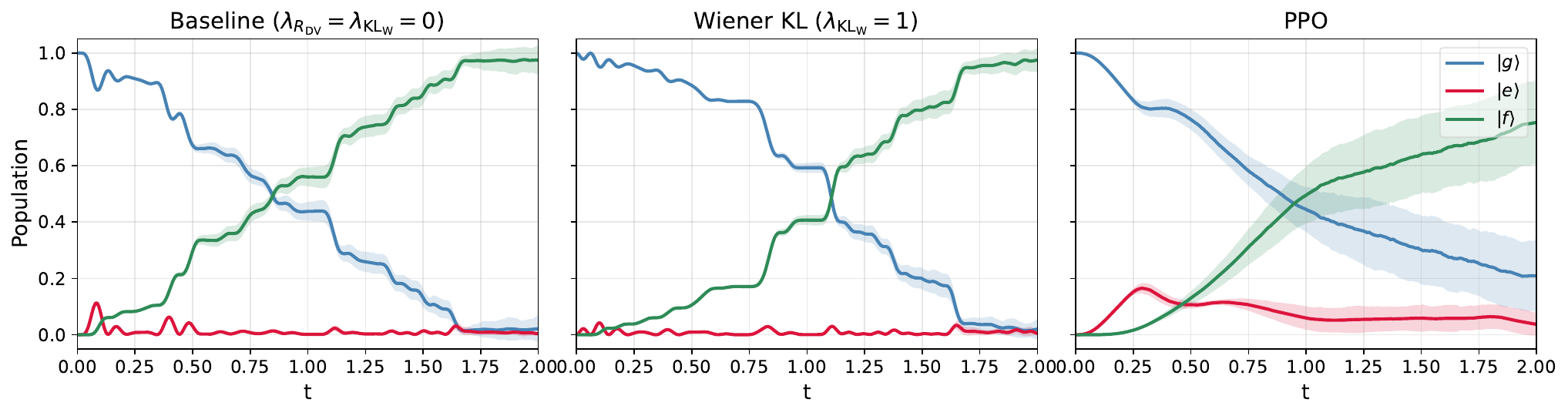}
  \caption{\textbf{STIRAP}, $\gamma T = 10$. SSE-trajectory populations (mean across 128 samples with $\pm 1\sigma$ bands) for baseline ($\lambda{=}0$), Wiener KL ($\lambda_{\klw}{=}1$), and PPO.}
  \label{fig:stirap-compare}
\end{figure}

\paragraph{Experiment-specific parameters.}
\begin{itemize}
\item $T=2.0$ in dimensionless units.
\item Decay rate sweep: $\gamma \in \{0.5, 5\}$ \rev{(}$\gamma T \in \{1,10\}$);
robustness evaluated at $\gamma_{\mathrm{test}}\in\{1,2,5,10,15\}$ from
policies trained at $\gamma=5$.
\item Total optimizer steps: $10\,000$.
\item Solver: ExpSplit
\item Discretized time grid: $256$ points per trajectory, snapshot every
$500$ steps; $64$ SSE samples per gradient step.
\item Regularizer weights: $\lambda_{\klw}\in\{0,1\}$, $\lambda_{\dvr}=0$,
$\lambda_{\mathrm{flu}}=0$.
\item Fourier controller: $n_{\mathrm{modes}}=16$, $\mathrm{init\_scale}=1.0$.
\item Seeds: $3$ per cell.
\end{itemize}

\subsection{Diamond system}
\label{app:diamond}

Table~\ref{tab:dr3_multiseed} reports the multi-seed fidelity 
statistics for the Wiener KL ablation at $\gamma = 2$ referenced in 
Section~\ref{sec:exp_dark}, and Table~\ref{tab:dr3_robustness} 
reports the $\gamma$-perturbation robustness experiment, with 
protocols trained at $\gamma_{\mathrm{train}} = 2$ and evaluated at 
six test noise levels spanning $\gamma_{\mathrm{test}} \in 
\{1, 1.5, 2, 2.5, 3, 5\}$.

\begin{table}[t]
\centering
\caption{\textbf{Diamond system: Wiener KL ablation at $\gamma=2$, $T=1$.}
Mean $\pm$ std across 3 seeds.}
\label{tab:dr3_multiseed}
\begin{tabular}{@{}lc@{}}
\toprule
method & $F_{\mathrm{exact}}$ \\
\midrule
baseline ($\lambda_{\klw}{=}\lambda_{\dvr}{=}0$)    & $0.665 \pm 0.017$ \\
Wiener KL $\lambda_{\klw}{=}0.5$   & $0.810 \pm 0.004$ \\
Wiener KL $\lambda_{\klw}{=}5$     & $\mathbf{0.834 \pm 0.003}$ \\
PPO                         & \rev{$0.605 \pm 0.007$} \\
\bottomrule
\end{tabular}
\end{table}

\begin{table}[t]
\centering
\caption{\textbf{Diamond system: $\gamma$-perturbation robustness.}
Mean $\pm$ std across 3 seeds.}
\label{tab:dr3_robustness}
\begin{tabular}{@{}ccccc@{}}
\toprule
$\gamma_{\mathrm{test}}$ & $\lambda_{\klw}{=}\lambda_{\dvr}{=}0$ & $\lambda_{\klw}{=}0.5$ & $\lambda_{\klw}{=}5$ & PPO \\
\midrule
$1.0$         & $0.808 \pm 0.010$ & $0.895 \pm 0.002$ & $\mathbf{0.903 \pm 0.002}$ & $0.772 \pm 0.001$ \\
$1.5$         & $0.732 \pm 0.014$ & $0.851 \pm 0.003$ & $\mathbf{0.868 \pm 0.003}$ & $0.681 \pm 0.002$ \\
$2.0$ (train) & $0.665 \pm 0.017$ & $0.810 \pm 0.004$ & $\mathbf{0.834 \pm 0.003}$ & $0.601 \pm 0.002$ \\
$2.5$         & $0.606 \pm 0.018$ & $0.771 \pm 0.004$ & $\mathbf{0.803 \pm 0.003}$ & $0.531 \pm 0.002$ \\
$3.0$         & $0.554 \pm 0.020$ & $0.734 \pm 0.004$ & $\mathbf{0.772 \pm 0.004}$ & $0.470 \pm 0.002$ \\
$5.0$         & $0.393 \pm 0.020$ & $0.607 \pm 0.006$ & $\mathbf{0.665 \pm 0.004}$ & $0.288 \pm 0.001$ \\
\bottomrule
\end{tabular}
\end{table}

\paragraph{Experiment-specific parameters.}
\begin{itemize}
\item $T=1.0$ in dimensionless units.
\item Training $\gamma=2.0$; robustness evaluated at
$\gamma_{\mathrm{test}}\in\{1,1.5,2,2.5,3,5\}$ from policies trained at
$\gamma=2$.
\item Total optimizer steps: $5000$.
\item Solver: EM
\item Discretized time grid: $256$ points per trajectory, snapshot every
$500$ steps; $256$ SSE samples per gradient step.
\item Regularizer weights: $\lambda_{\klw}\in\{0,0.5,5\}$, $\lambda_{\dvr}=0$,
$\lambda_{\mathrm{flu}}=0.001$.
\item Fourier controller: $n_{\mathrm{modes}}=20$, $\mathrm{init\_scale}=0.1$.
\item Seeds: $3$ per cell.
\end{itemize}

\subsection{Four-qubit chain asymmetry sweep}
\label{app:chain4_full}

Table~\ref{tab:chain4_full} reports the complete five-point
asymmetry sweep ($\rho\in\{1,1.5,2,4,8\}$) summarized in
Section~\ref{sec:exp_chain}. Means and standard deviations are taken
over three seeds for $\rho\in\{2,4,8\}$; $\rho=1$ and $\rho=1.5$ are
reported on a single seed as control points. The single
representative trajectories shown in
Figure~\ref{fig:chain4_asym_pops} are taken from the seed with the
highest \emph{final-step} exact-Lindblad fidelity at each $\rho$;
per-seed best fidelities are reported in the lower half of
Table~\ref{tab:chain4_full}. Due to high noise values, training was unstable so we implemented a fluence warm-up of $2000$ steps.

\begin{table}[h]
\centering
\caption{\textbf{Chain-4 asymmetry sweep:} maximum exact-Lindblad
fidelity for excitation transfer site-0 to site-3, by
edge-to-interior dephasing ratio $\rho$. \emph{Top:} mean$\pm$std
over three seeds where applicable. \emph{Bottom:} per-seed best
fidelities for the multi-seed rows.}
\label{tab:chain4_full}
\begin{tabular}{lccc}
\toprule
$\rho$ & Baseline $F$ & Drift-variance ($\lambda_{\dvr}=0.02$) $F$ & $\Delta$ \\
\midrule
1.0 (uniform)        & 0.71  & 0.71  & $\phantom{+}0.0$\,pp \\
1.5                  & 0.69  & 0.69  & $\phantom{+}0.0$\,pp \\
2.0                  & $0.64\pm0.01$ & $0.62\pm0.03$ & $-2.4$\,pp \\
4.0                  & $0.48\pm0.02$ & $0.52\pm0.13$ & $+4.5$\,pp \\
8.0 & $0.22\pm0.04$ & $0.42\pm0.10$ & $+19.9$\,pp \\
\bottomrule
\end{tabular}

\vspace{0.75em}

\begin{tabular}{llccc}
\toprule
$\rho$ & policy & seed 0 & seed 1 & seed 2 \\
\midrule
2.0 & baseline        & 0.637 & 0.636 & 0.651 \\
2.0 & drift-variance  & 0.606 & 0.591 & 0.654 \\
4.0 & baseline        & 0.482 & 0.503 & 0.454 \\
4.0 & drift-variance  & 0.609 & 0.589 & 0.374 \\
8.0 & baseline        & 0.193 & 0.209 & 0.262 \\
8.0 & drift-variance  & 0.501 & 0.456 & 0.306 \\
\bottomrule
\end{tabular}
\end{table}

The wide drift-variance spread at $\rho=4$ and $\rho=8$
(std~$\approx0.10$--$0.13$) is concentrated in seed~2, which finds
a qualitatively different basin in both cases. The mean improvement
over the baseline holds at all asymmetry levels $\rho\geq 4$ even
including this seed.

\paragraph{Experiment-specific parameters.}
\begin{itemize}
\item $T=3.0$ in dimensionless units; rates are
in inverse-time units of the same scale.
\item Per-site decoherence rates $\kappa_{\mathrm{int}}=0.3$,
$\gamma_{\mathrm{int}}=0.05$ at all interior sites, scaled by
$\rho$ on the two edge sites
($\kappa_0\!=\!\kappa_3\!=\!\rho\,\kappa_{\mathrm{int}}$ and similarly
for $\gamma$). The scale is chosen so that
$T\,\kappa_{\mathrm{int}}\!\approx\!1$, placing
decoherence on the same timescale as the protocol.
\item Asymmetry sweep: $\rho\in\{1,1.5,2,4,8\}$.
\item Total optimizer steps: $10\,000$.
\item Solver: ExpSplit
\item Discretized time grid: $256$ points per trajectory, snapshot
every $500$ steps; $64$ SSE samples per gradient step.
\item Regularizer weights: $\lambda_{\klw}=0$,
$\lambda_{\dvr}\in\{0,0.02\}$,
$\lambda_{\mathrm{flu}}=0.005$.
\item Fluence warm-up: $2000$ steps (penalty held at zero, then
enabled).
\item Fourier controller: $n_{\mathrm{modes}}=16$,
$\mathrm{init\_scale}=1.0$.
\item Seeds: $1$ at $\rho\in\{1,1.5\}$, $3$ at $\rho\in\{2,4,8\}$.
\end{itemize}

\subsection{Hardware-calibrated six-qubit chain experiment}
\label{app:ibm_chain}
Tables~\ref{tab:ibm_chain_rates} and~\ref{tab:ibm_chain_full}
report the per-site coherence times and the multi-seed fidelity
statistics for the six-qubit excitation-transfer experiment of
Section~\ref{sec:exp_chain}. The calibration source is 
\texttt{ibm\_kingston\_calibrations\_2026-05-03T11\_33\_48Z.csv} 
(supplementary materials), retrieved from the IBM Quantum Platform 
at 11:33:48 UTC on 3 May 2026 for processor 
\texttt{ibm\_kingston}~\citep{ibmquantum}. IBM Quantum calibrations 
refresh approximately daily; the snapshot used here is fixed, so the 
per-site rates in Table~\ref{tab:ibm_chain_rates} are reproducible 
exactly from the supplementary file but will differ from any newer 
snapshot. Per-qubit $T_1$ and $T_2$ are converted to Lindblad rates 
via
\[
  \gamma_i = 1/T_1^{(i)}, \qquad
  \kappa_i = \tfrac{1}{2}\!\max\!\Big(0,\;
              1/T_2^{(i)} - 1/(2 T_1^{(i)})\Big),
\]
yielding two per-site Lindblad operators per qubit
($\sqrt{\gamma_i}\,\sigma_-^{(i)}$ amplitude damping and
$\sqrt{\kappa_i}\,\sigma_z^{(i)}$ dephasing). The drive is a sum of
five nearest-neighbor exchange couplings,
$H_{\mathrm{ctrl}}(t)=\sum_{i=0}^{4} g_{i,i+1}(t)\,
\tfrac{1}{2}(\sigma_x^{(i)}\sigma_x^{(i+1)}+\sigma_y^{(i)}\sigma_y^{(i+1)})$,
with no fixed exchange or local drive. Training used a fluence 
warm-up for stability.

\begin{figure}[t]
    \centering
    \includegraphics[width=\linewidth]{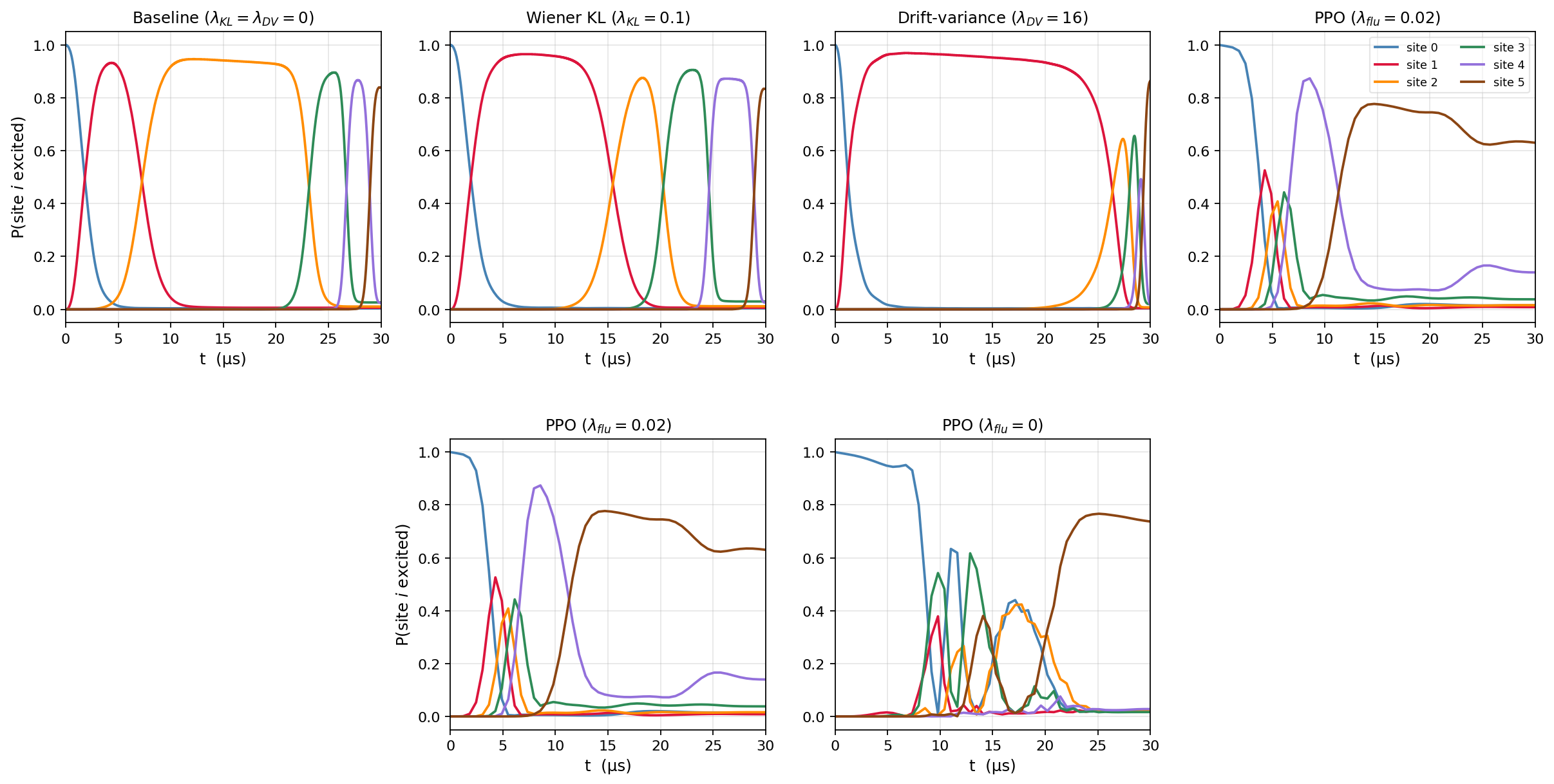}
    \caption{\textbf{IBM Kingston q14$\to$q9 evolutions.} Per-site
    excitation populations
    $\langle\psi(t)|(I-\sigma_z^{(i)})/2|\psi(t)\rangle$
    (i.e., the probability that site $i$ holds the excitation).
    \emph{Top row}: the four policies of
    Table~\ref{tab:ibm_chain_full} at matched squared-fluence penalty
    $\lambda_{\mathrm{flu}}=0.02$\rev{:} baseline
    ($\lambda_{\klw}{=}\lambda_{\dvr}{=}0$), Wiener KL
    ($\lambda_{\klw}{=}0.1$), drift-variance
    ($\lambda_{\dvr}{=}16$), and PPO. Drift-variance reaches the
    highest final-time excitation on q9.
    \emph{Bottom row}: PPO compared at matched fluence penalty
    ($\lambda_{\mathrm{flu}}=0.02$) and with the
    fluence penalty removed entirely ($\lambda_{\mathrm{flu}}=0$).}
    \label{fig:chain6_evolutions}
\end{figure}

\begin{table}[h]
\centering
\caption{Per-site coherence times for the IBM Kingston q14$\to$q9
chain, with derived $\gamma+\kappa$ rate. Site index runs from
source (left, q14) to target (right, q9).}
\label{tab:ibm_chain_rates}
\begin{tabular}{ccccc}
\toprule
Site & Qubit & $T_1\,(\mu\mathrm{s})$ & $T_2\,(\mu\mathrm{s})$ & $\gamma+\kappa\,(\mu\mathrm{s}^{-1})$ \\
\midrule
0 & q14 & 147.5 & 191.2 & 0.0077 \\
1 & q13 & 402.9 & 352.0 & 0.0033 \\
2 & q12 & 355.1 & 272.8 & 0.0039 \\
3 & q11 & 402.7 &  93.3 & 0.0072 \\
4 & q10 & 346.5 &  17.0 & 0.0316 \\
5 &  q9 & 137.8 &  61.8 & 0.0135 \\
\bottomrule
\end{tabular}
\end{table}

\begin{table}[h]
\centering
\caption{\textbf{IBM Kingston q14$\to$q9:} maximum exact-Lindblad
fidelity for excitation transfer, $T=30\,\mu\mathrm{s}$,
linear fluence penalty $\lambda_{\mathrm{flu}}=0.02$. Mean $\pm$ std
across 3 seeds where applicable. Bold marks the best. The bottom block reports two control experiments (overdriven KL, unconstrained PPO)\rev{; Appendix~\ref{app:fluence_ablation} reports all methods without the fluence penalty}.}
\label{tab:ibm_chain_full}
\begin{tabular}{@{}lc@{}}
\toprule
method & $F_{\mathrm{exact}}$  \\
\midrule
baseline ($\lambda_{\klw}{=}\lambda_{\dvr}{=}0$, $\lambda_{\mathrm{flu}}=0.02$)        & $0.837 \pm 0.004$ \\
Wiener KL ($\lambda_{\klw}{=}0.1$, $\lambda_{\mathrm{flu}}=0.02$) & $0.831 \pm 0.004$ \\
Drift-variance ($\lambda_{\dvr}{=}16$, $\lambda_{\mathrm{flu}}=0.02$)                  & $\mathbf{0.863 \pm 0.002}$ \\
PPO ($\lambda_{\mathrm{flu}}=0.02$)                        & $0.604 \pm 0.041$ \\
\midrule
Wiener KL ($\lambda_{\klw}{=}16$, single seed)          & $0.427$ \\
\rev{PPO ($\lambda_{\mathrm{flu}}=0$, ten seeds)} & \rev{$0.640 \pm 0.031$}\\
\bottomrule
\end{tabular}
\end{table}

\paragraph{Experiment-specific parameters.}
\begin{itemize}
\item $T=30\,\mu\mathrm{s}$ in physical units;
rates above are in $\mu\mathrm{s}^{-1}$.
\item Initial state $\ket{100000}$ (excitation at site 0, q14),
target $\ket{000001}$ (excitation at site 5, q9).
\item Total optimizer steps: $3000$.
\item Solver: ExpSplit
\item Discretized time grid: $1024$ points per trajectory, snapshot
every $250$ steps; $64$ SSE samples per gradient step.
\item Regularizer weights: $\lambda_{\mathrm{flu}}=0.02$;
$\lambda_{\dvr}\in\{0,16\}$;
$\lambda_{\klw}\in\{0,0.1\}$, with the non-zero value chosen so
that the average loss contribution of the KL term matches that of
the drift-variance term at $\lambda_{\dvr}=16$ (we also tried
$\lambda_{\klw}=16$, reported as a single-seed control).
\item Fluence warm-up: $500$ steps.
\item Fourier controller: $n_{\mathrm{modes}}=16$,
$\mathrm{init\_scale}=1.0$.
\item Seeds: $3$ for each of baseline, $\lambda_{\dvr}=16$, and
$\lambda_{\klw}=0.1$.
\end{itemize}

\section{PPO baseline configuration}
\label{app:ppo}

We use the constrained-RL formulation of \citet{ernst2025} as our
trajectory-based RL baseline.

\paragraph{Reference implementation.} We use the authors'
\texttt{train\_ppo.py} from the released codebase \citep{ernst2025}
together with their PPO implementation, wrapping our QMaxCal system
modules through \rev{the codebase's} generic environment interface. The reward
function follows Eq.~2 of \citet{ernst2025}.

\paragraph{Architecture and core hyperparameters.} Separate
actor-critic MLPs with hidden width $256$ and ReLU6 activations;
Gaussian policy. PPO with clip $\epsilon=0.2$, GAE $\lambda=0.95$,
discount $\gamma_{\rm RL}=0.99$, max gradient norm $0.5$, $4$
minibatches per update, no entropy bonus. Adam optimizer with linear
learning-rate annealing. Training runs $50{,}000$ update steps with
$16$ parallel environments and $50$ piecewise-constant action samples
per episode ($4 \times 10^7$ environment interactions per run),
smoothed by a Gaussian kernel of length $25$. The constrained-RL
maximum ODE-solver step is $N_{\max}^{\rm sim}=1000$ with penalty
$25$ when exceeded. For the IBM Kingston chain-6 system only, we
chunk training into blocks of $5000$ updates (saving an
intermediate snapshot at each chunk boundary) and cap each run at
$3$\,h of GPU wall time; if the cap is hit, the most recent
snapshot is used. For reference, the QMaxCal training runs for the six-qubit chain
completed in under $2$\,h of GPU wall time.

\paragraph{Hyperparameter sweep.} We swept the PPO hyperparameters in two sequential phases per benchmark.

\emph{Phase A:} $N_{\max}^{\rm sim} \times \lambda_{\rm flu}$ on a
$\{500, 1000, 2000, 4000\} \times \{0.001, 0.01, 0.1\}$ grid for
amplitude damping, where $N_{\max}^{\rm sim}$ is the \rev{main hyperparameter studied in} Fig.~1 of \citet{ernst2025} and $\lambda_{\rm flu}$ is a
quadratic fluence penalty $\lambda_{\rm flu}\!\int_0^T |\Omega|^2 dt$
that we add to \rev{the PPO} reward (distinct from \rev{the} linear
pulse-area penalty $w_A \!\int|\Omega|dt$ in Eq.~2; the quadratic
form matches QMaxCal's regularizer and places PPO on the same
fluence--fidelity Pareto frontier; we set \rev{the linear weight} $w_A=0$).
For STIRAP and the diamond system we used the reduced Phase A grid
$\{1000, 4000\} \times \{0.001, 0.01, 0.1\}$ ($6$ cells), since the
amplitude-damping sweep showed $N_{\max}^{\rm sim}$ to be flat in
the relevant range. For the IBM Kingston chain-6 we omit Phase A
entirely: $\lambda_{\rm flu}$ is fixed at $0.02$ to match the
QMaxCal regularizer strength, and $N_{\max}^{\rm sim}=1000$ is
inherited from the converged Phase-A choice on the smaller systems.

\emph{Phase B:} at the best Phase-A cell, the constraint-weight
ablation from Fig.~14 of \citet{ernst2025}: smoothness penalty
$\in \{0, 0.001, 0.01\}$, Gaussian-kernel std $\in \{2, 4, 8\}$,
and fidelity weight $w_F \in \{1, 5\}$ ($18$ cells). For the IBM
Kingston chain-6 we further restrict Phase B to a $2\times2$ grid
over smoothness $\in \{0, 0.001\}$ and kernel std $\in \{2, 4\}$
with $w_F$ fixed at $1$, motivated by GPU-budget constraints (each
matched-fluence chain-6 PPO run takes $\sim\!2.5$\,h vs.\ $\sim\!1$\,h
on the smaller systems).

This yields $12 + 18 = 30$ configurations on amplitude damping,
$6 + 18 = 24$ on STIRAP and the diamond system, and $0 + 4 = 4$ on
the IBM Kingston chain-6.

\emph{Seeds and selection.} Sweep runs used a single seed; the
final reported numbers come from $4$ seeds on amplitude damping,
$3$ seeds on STIRAP, the diamond system, and the IBM Kingston
chain-6 (per Tables~\ref{tab:ampdamp_fidelity_full},
\ref{tab:stirap_multiseed}, \ref{tab:dr3_multiseed},
\ref{tab:ibm_chain_full}). Selection criterion: highest
exact-Lindblad fidelity $F$ on the deterministic-rollout snapshot.

\paragraph{Per-system overrides.} Phase A converged to
$(N_{\max}^{\rm sim}, \lambda_{\rm flu}) = (1000, 0.001)$ for
amplitude damping, STIRAP, and the diamond system; for IBM
Kingston, $\lambda_{\rm flu}$ was set to $0.02$ a priori. Phase B
winners differed per benchmark:

\begin{center}
\begin{tabular}{@{}lcccc@{}}
\toprule
benchmark & smoothness & kernel std & $w_F$ & $\lambda_{\rm flu}$ \\
\midrule
amp.\ damping     & $0$     & $2$ & $5$ & $0.001$ \\
STIRAP            & $0.001$ & $2$ & $5$ & $0.001$ \\
diamond           & $0.001$ & $2$ & $1$ & $0.001$ \\
IBM Kingston q14$\to$q9 & $0.001$ & $4$ & $1$ & $0.02$ \\
\bottomrule
\end{tabular}
\end{center}

\paragraph{Fluence-penalty sanity check.} To confirm that the 
quadratic fluence penalty $\mathrm{Flu} = \sum_a \int_0^T |u_a^{(\theta)}(t)|^2\,
dt$ 
is not handicapping PPO, we re-trained the best Phase-B 
configuration on all four benchmarks with $\lambda_{\rm flu}=0$ (single seed each\rev{ on amplitude damping and STIRAP, three seeds on the diamond system and ten on IBM Kingston}). Removing the penalty did not improve PPO on amplitude damping and STIRAP: fidelity was 
unchanged within seed noise. On the diamond system, pulse fluence grew by \rev{$28\times$ ($116 \to 3277$)} while fidelity dropped (\rev{$F: 0.605 \pm 0.007 \to 0.530 \pm 0.058$}). On the IBM Kingston six-qubit chain, removing the penalty raised fidelity from $0.604$ to \rev{$0.640 \pm 0.031$ over ten seeds} (Table~\ref{tab:ibm_chain_full}), while the fluence \rev{increased} \rev{$68\times$ ($8.0 \to 543$, means over seeds)}. \rev{The matched-fluence constraint is therefore not the cause of the gap to QMaxCal. We report both numbers in the main text (Table~\ref{tab:ibm_chain}); the matched-fluence number remains the like-for-like comparison with QMaxCal}, which uses $\lambda_{\rm flu}=0.02$.
\rev{
\subsection{Noise dependence of the PPO gap}
\label{app:ppo_noise}
The gap between PPO and the unregularized baseline depends on the noise strength. On amplitude damping (Table~\ref{tab:ampdamp_fidelity_full}), PPO matches the baseline for $\gamma T \leq 1$ ($0.9943$ vs.\ $0.9942$, $0.9856$ vs.\ $0.9831$, $0.9773$ vs.\ $0.9772$) and falls behind only for $\gamma T \geq 2$ ($1.8$\,pp at $\gamma T = 5$). On the diamond system (Table~\ref{tab:ppo_noise_diamond}), PPO without fluence penalty matches or exceeds the baseline at low noise and falls below it only at the benchmark noise, while the Wiener KL leads at every level. On both systems the gap grows with the decoherence strength. It is not explained by the tuning budget (30, 24, 24 and 4 sweep configurations per benchmark; PPO runs used the majority of this project's compute) or by the matched-fluence constraint (see the sanity check above).

\begin{table}[h]
\centering
\color{revcolor}
\caption{Diamond system: best-checkpoint exact-Lindblad fidelity at three training noise levels; each policy is trained and evaluated at the same $\gamma$. Mean $\pm$ std ($n-1$) over three seeds, two for PPO at $\gamma = 0.5$ (one run diverged). At $\gamma = 2$ the main-text values ($0.665$, $0.834$) instead select the controller with the lowest training loss.}
\label{tab:ppo_noise_diamond}
\begin{tabular}{@{}cccc@{}}
\toprule
$\gamma$ & baseline & Wiener KL ($\lambda_\klw = 5$) & PPO ($\lambda_{\mathrm{flu}} = 0$) \\
\midrule
$0.1$ & $0.905 \pm 0.000$ & $\mathbf{0.961 \pm 0.000}$ & $0.932 \pm 0.004$ \\
$0.5$ & $0.740 \pm 0.033$ & $\mathbf{0.919 \pm 0.000}$ & $0.786 \pm 0.006$ \\
$2$ (benchmark) & $0.670 \pm 0.016$ & $\mathbf{0.837 \pm 0.001}$ & $0.530 \pm 0.058$ \\
\bottomrule
\end{tabular}
\end{table}
}

\section{\rev{Ablations}}
\label{app:smoothness_ablation}
\rev{\subsection{Smoothness regularization is not Wiener KL}}
We rule out the hypothesis that Wiener KL acts as an implicit smoothness penalty on the controls $u_a^{(\theta)}(t)$. The derivative penalties of \citet{abdelhafez2019} (C4 and C5) take the form
\[
  \lambda_{d_1}\!\int_0^T \!|u'(t)|^2\,dt \;+\;
  \lambda_{d_2}\!\int_0^T \!|u''(t)|^2\,dt
\]
\rev{We} add them to QMaxCal at two weight settings on amplitude damping 
($\gamma T=2$) and STIRAP ($\gamma T=10$), with smoothness terms 
warmed up over the first $1500$ optimizer steps.

Table~\ref{tab:smooth_ablation} shows that adding C4/C5 \rev{lowers fidelity} in every cell and never recovers the Wiener-KL-only 
result. On STIRAP the heavier weight collapses fidelity by 
$16$\,pp; combining it with Wiener KL adds at most $1$\,pp over 
smoothness alone, confirming that Wiener KL cannot rescue a 
trajectory already constrained by pulse-derivative penalties. \rev{Reported numbers are single-seed. The $16$\,pp drop on STIRAP is far larger than the seed spread of the multi-seed STIRAP runs (std $\leq 0.001$, Table~\ref{tab:stirap_multiseed}).}

The two regularizers act on physically distinct objects: Wiener KL 
acts on the dissipation-channel drift 
$\alpha_k(t)=2\,\mathrm{Re}\langle\psi|L_k|\psi\rangle$ in 
trajectory space, driving it toward $\ker L$; C4/C5 act on 
derivatives of the controls in pulse space, agnostic to the 
trajectory's relationship to the noise kernel. They cannot 
substitute for each other.

\begin{table}[t]
\centering
\caption{\textbf{Smoothness ablation.} Single seed; smoothness 
warmed up over $1500$ steps. \emph{Top:} amplitude damping 
($\gamma T=2$, $\lambda_{\klw}=5$). \emph{Bottom:} STIRAP 
($\gamma T=10$, $\lambda_{\klw}=1$). Bold marks the highest 
fidelity per block.}
\label{tab:smooth_ablation}
\begin{tabular}{@{}llc@{}}
\toprule
system & method & $F_{\rm exact}$ \\
\midrule
amp.\ damp. & baseline & $0.972$ \\
amp.\ damp. & Wiener KL & $\mathbf{0.981}$ \\
amp.\ damp. & smooth only ($10^{-3},10^{-5}$) & $0.932$ \\
amp.\ damp. & smooth $+$ Wiener KL ($10^{-3},10^{-5}$) & $0.941$ \\
\midrule
STIRAP & baseline & $0.978$ \\
STIRAP & Wiener KL & $\mathbf{0.979}$ \\
STIRAP & smooth only ($10^{-3},10^{-5}$) & $0.819$ \\
STIRAP & smooth $+$ Wiener KL ($10^{-3},10^{-5}$) & $0.822$ \\
STIRAP & smooth only ($10^{-4},10^{-6}$) & $0.914$ \\
STIRAP & smooth $+$ Wiener KL ($10^{-4},10^{-6}$) & $0.919$ \\
\bottomrule
\end{tabular}
\end{table}

\paragraph{Parameters.} Amplitude damping: $T=1$, $\gamma=2$, 
$5000$ steps, $256$ SSE samples, $n_{\rm modes}=20$. 
STIRAP: $T=2$, $\gamma=5$, $10\,000$ steps, $64$ SSE samples, 
$n_{\rm modes}=16$.

\rev{
\subsection{Removing the fluence penalty}
\label{app:fluence_ablation}
We trained all four IBM Kingston configurations with $\lambda_{\mathrm{flu}} = 0$ (Table~\ref{tab:fluence_ablation}). Removing the fluence penalty changes no QMaxCal configuration by more than $0.01$, so the gain of $\dvr$ comes from the path-space regularizer rather than from the fluence term; only PPO is noticeably affected ($0.604 \to 0.640$). $\dvr$ remains the best method with and without the fluence penalty, consistent with the analysis of Section~\ref{sec:zero_sets} for this dephasing-dominated chain. STIRAP already runs with $\lambda_{\mathrm{flu}} = 0$ (Appendix~\ref{app:stirap}).

\begin{table}[h]
\centering
\color{revcolor}
\caption{IBM Kingston q14$\to$q9: fidelity with ($\lambda_{\mathrm{flu}} = 0.02$, Table~\ref{tab:ibm_chain_full}) and without ($\lambda_{\mathrm{flu}} = 0$) the fluence penalty. Mean $\pm$ std ($n-1$) over three seeds (ten for PPO without the penalty).}
\label{tab:fluence_ablation}
\begin{tabular}{@{}lcc@{}}
\toprule
method & $\lambda_{\mathrm{flu}} = 0.02$ & $\lambda_{\mathrm{flu}} = 0$ \\
\midrule
baseline ($\lambda_\klw = \lambda_\dvr = 0$) & $0.837 \pm 0.004$ & $0.836 \pm 0.014$ \\
Wiener KL ($\lambda_\klw = 0.1$) & $0.831 \pm 0.004$ & $0.834 \pm 0.019$ \\
Drift-variance ($\lambda_\dvr = 16$) & $\mathbf{0.863 \pm 0.002}$ & $\mathbf{0.871 \pm 0.006}$ \\
PPO & $0.604 \pm 0.041$ & $0.640 \pm 0.031$ \\
\bottomrule
\end{tabular}
\end{table}

Table~\ref{tab:fluence_sweep} sweeps the fluence weight with both path-space regularizers off. Fluence regularization alone peaks at $F = 0.851$ ($\lambda_{\mathrm{flu}} = 0.01$), below drift-variance ($0.863$--$0.871$, Table~\ref{tab:fluence_ablation}), and fidelity collapses for $\lambda_{\mathrm{flu}} \geq 1$. We report final-step values, because for $\lambda_{\mathrm{flu}} \geq 0.05$ the highest-fidelity checkpoint occurs at step $500$, before the penalty has reduced the fluence.

\begin{table}[H]
\centering
\color{revcolor}
\small
\setlength{\tabcolsep}{3pt}
\caption{IBM Kingston q14$\to$q9 with fluence regularization only ($\lambda_\klw = \lambda_\dvr = 0$): final-step exact-Lindblad fidelity (mean over three seeds) and final fluence (range over seeds).}
\label{tab:fluence_sweep}
\begin{tabular}{@{}lcccccccccc@{}}
\toprule
$\lambda_{\mathrm{flu}}$ & $0$ & $0.005$ & $0.01$ & $0.02$ & $0.05$ & $0.1$ & $0.2$ & $0.5$ & $1$ & $2$ \\
\midrule
$F$ & $0.836$ & $0.826$ & $\mathbf{0.851}$ & $0.837$ & $0.798$ & $0.755$ & $0.712$ & $0.633$ & $0.000$ & $0.000$ \\
fluence & $144$--$189$ & $11$--$17$ & $5.0$--$5.2$ & $3.6$--$3.8$ & $2.4$ & $1.8$ & $1.4$ & $1.2$ & $0.0$ & $0.0$ \\
\bottomrule
\end{tabular}
\end{table}
}

\section{SSE solver schemes and scaling behavior}
\label{app:scaling}

This appendix documents the two diffusive-SSE integrators used to
optimize QMaxCal protocols, characterizes how each one scales with
system size $N$, and contrasts both with the Lindblad master-equation
solver (\texttt{qiskit-dynamics}~\citep{puzzuoli2023qiskit} on top of \texttt{diffrax}~\citep{kidger2021neural} Dopri5)
used as ground truth. Our SSE-integrator implementations build on
\texttt{dynamiqs}~\citep{guilmin2025dynamiqs}, a JAX-based library for
differentiable open-quantum-system simulation; we extended its
diffusive-SSE solvers with the Girsanov KL and drift-variance
loss terms of Section~\ref{sec:methods}.

\subsection{The two SSE integrators}
\label{app:scaling:integrators}

We solve the diffusive Stochastic Schr\"odinger Equation
\begin{equation}
\begin{split}
  d|\psi\rangle = \Big[ &-iH\,dt - \tfrac{1}{2}\sum_k L_k^\dagger L_k\,dt
                 + \tfrac{1}{2}\sum_k\langle L_k+L_k^\dagger\rangle L_k\,dt \\
                &- \tfrac{1}{8}\sum_k\langle L_k+L_k^\dagger\rangle^2\,dt\Big]
                 |\psi\rangle
                 + \sum_k\!\big(L_k - \tfrac{1}{2}\langle L_k+L_k^\dagger\rangle\big)|\psi\rangle\, dW_k
\end{split}
\end{equation}
with two integrators, summarized in Table~\ref{tab:em_vs_expsplit}.

\paragraph{Euler--Maruyama (EM).}
A \rev{weak-order-1} stochastic Euler step that applies $H$ and
the $\{L_k\}$ to $|\psi\rangle$ as matrix--vector products. EM never
instantiates a dense propagator, so memory stays linear in $\dim$.
\rev{The per-step error is $O(dt^2\|H_{\mathrm{eff}}\|^2)$, accumulating to} $O(dt\cdot\|H_{\mathrm{eff}}\|^2)$ \rev{over the protocol}, so
large-norm Hamiltonians require proportionally finer time
discretization.

\paragraph{Exponential split-step (ExpSplit).}
Splits each timestep into a deterministic half, solved exactly via
$\exp(-iH_{\mathrm{eff}}\,dt)|\psi\rangle$ with $H_{\mathrm{eff}} = H -
i\tfrac{1}{2}\sum_k L_k^\dagger L_k$, and an Euler-style stochastic
correction. The matrix exponential makes the deterministic half
Hamiltonian-exact and removes the $\|H\|$-dependence of EM's
truncation error, but instantiates a dense $\dim\times\dim$ propagator
at every step, costing $O(\dim^3)$ compute and $O(\dim^2)$ memory.

\begin{table}[h]
\centering
\caption{Per-step cost and error scaling. $n_{\mathrm{trajs}}$ is the
number of stochastic samples; $\dim = 2^N$.}
\label{tab:em_vs_expsplit}
\begin{tabular}{@{}lccc@{}}
\toprule
solver & per-step compute & per-step memory & Hamiltonian error \\
\midrule
EM       & $O(\dim^2 \cdot n_{\mathrm{trajs}})$ & $O(\dim \cdot n_{\mathrm{trajs}})$ & $O(dt^2 \|H\|^2)$ \\
ExpSplit & $O(\dim^3)$                          & $O(\dim^2)$                         & exact \\
\bottomrule
\end{tabular}
\end{table}

The Lindblad reference solver propagates the full density matrix
$\rho(t)$ via adaptive Dopri5 on the Liouvillian, costing $O(\dim^3)$
compute and $O(\dim^2)$ memory per adaptive step.

\subsection{Scaling with $N$}
\label{app:scaling:oom}

We benchmarked all three solvers on random $N$-qubit XX$+$YY chains
(per-site $\sigma_-$ and $\sigma_z$ Lindblad channels, random
nearest-neighbor exchange, $T_{\mathrm{final}}=1$, Fourier control
pulses with $n_{\mathrm{modes}}=4$). All runs used a single 40\,GB
A100. Wall-clock and peak GPU memory are reported in
Table~\ref{tab:scaling_wallclock}.

\begin{table}[h]
\centering
\caption{Cached forward-pass wall-clock and peak GPU memory for the
three solvers on random $N$-qubit chains; same controller, same
Lindblad operators, $n_{\mathrm{trajs}}=256$,
$n_{\mathrm{time\text{-}steps}}=2048$, single 40\,GB A100.}
\label{tab:scaling_wallclock}
\begin{tabular}{@{}cccccc@{}}
\toprule
$N$ & $\dim$ & Lindblad call & Lindblad peak GPU & EM call & ExpSplit call \\
\midrule
4  & 16    & ---         & ---         & $0.1$\,s          & $1.2$\,s \\
6  & 64    & $3.9$\,s    & $82$\,MiB   & $0.2$\,s          & $1.4$\,s \\
8  & 256   & $4.3$\,s    & $130$\,MiB  & $0.5$\,s          & $2.7$\,s \\
10 & 1024  & $21.8$\,s   & $1.85$\,GiB & $3.6$\,s          & $15.1$\,s \\
12 & 4096  & \multicolumn{2}{c}{\textbf{GPU-OOM}} & $\mathbf{45.8}$\,s & JIT did not complete in 2\,h \\
\bottomrule
\end{tabular}
\end{table}

\rev{Two observations follow.} First, the two $O(\dim^3)$ solvers \rev{fail} at the same scale: by $N=12$, the dense $\dim\times\dim$
matrices needed for density-matrix propagation (Lindblad) and for the
matrix exponential (ExpSplit) exceed practical GPU memory and JIT-compile budgets\rev{;} the Lindblad solver runs out of GPU memory and
the ExpSplit solver fails to finish JIT-lowering within a $2$\,h
budget. EM's $O(\dim\cdot n_{\mathrm{trajs}})$ state representation \rev{extends the reachable size of the forward pass by} roughly two qubits, with rapidly growing wall-clock. Second, at the system sizes that are reachable for all
three (up through $N=10$), EM and ExpSplit run several times faster
than the Lindblad reference on identical hardware\rev{. For EM this is} consistent with the quadratic state-space advantage of trajectory representations
discussed in Appendix~\ref{app:quantum_primer}.

\rev{
\subsection{Forward simulation versus gradient-based training}
\label{app:scaling:train}
The $12$-qubit result in Table~\ref{tab:scaling_wallclock} measures forward-simulation capacity: the SSE forward pass runs at $\dim = 4096$, where the density-matrix reference solver runs out of memory. It does not measure training capacity. Gradient-based training backpropagates through the solver, which requires storing the full forward trajectory, and our current implementation could not train at this size. The six-qubit chain of Section~\ref{sec:exp_chain} is the largest system trained in this paper; we did not determine the largest trainable size. The regularizers add essentially no cost on top of the solver, because the drifts $\alpha_k$ are already computed at every integrator step. Both regularizers are expectations over forward SSE samples, so they can also serve as objectives for non-gradient optimizers such as RL or evolutionary strategies, which we have not tested; for these, the forward-simulation limit is the relevant ceiling.
}

\subsection{When EM is too coarse: Hamiltonian-amplitude sensitivity}
\label{app:scaling:em_amp}

EM's \rev{accumulated} error scales as $O(dt\cdot\|H_{\mathrm{eff}}\|^2)$, so EM
becomes inaccurate when $\|H_{\mathrm{eff}}\|$ is large at fixed
time-grid resolution. We illustrate with a controlled experiment at
$N=10$: fixed seed, $n_{\mathrm{trajs}}=1024$ (Monte Carlo noise floor
$\sim\!0.031$ in Frobenius distance to Lindblad), random-Fourier
pulse amplitude varied over $\{1, 3, 10\}$ (peak
$\|H\|_{\mathrm{op}}\!\approx\!17, 35, 105$), and
$n_{\mathrm{time\text{-}steps}}$ swept over a 256$\times$ range.

\begin{table}[h]
\centering
\caption{EM Frobenius distance to the Lindblad ground truth, $N=10$.
Rows: pulse amplitude. Columns: time-discretization. Bold cells
indicate bias dominating the Monte Carlo floor of $\sim\!0.031$.}
\label{tab:em_accuracy_amp}
\begin{tabular}{@{}cccccccc@{}}
\toprule
pulse amp & $\|H\|_{\mathrm{op}}$ & $\mathrm{nts}{=}256$ & 1024 & 4096 & 16384 & 65536 \\
\midrule
1   & $17$  & $0.039$ & $0.028$ & $0.029$ & $0.024$ & $0.032$ \\
3   & $35$  & $\mathbf{0.121}$ & $0.038$ & $0.031$ & $0.028$ & $0.031$ \\
10  & $105$ & $\mathbf{0.629}$ & $\mathbf{0.204}$ & $\mathbf{0.052}$ & $0.033$ & $0.030$ \\
\bottomrule
\end{tabular}
\end{table}

Three observations: (i) at amplitude $1$ EM is converged at
$n_{\mathrm{time\text{-}steps}}=256$ and refining the grid 256-fold
does not change the result; (ii) at amplitude $3$ a $4\times$ refinement
reaches the noise-limited regime; (iii) at amplitude $10$ EM bias
dominates the Frobenius error until $n_{\mathrm{time\text{-}steps}}\geq
16384$, with the 256-step run producing a Frobenius error of $0.63$\rev{, twenty times the Monte Carlo floor}. The scaling tracks the theoretical prediction
$\mathrm{bias}\sim dt\cdot\|H_{\mathrm{eff}}\|^2\cdot T$: a $6\times$
growth in $\|H\|_{\mathrm{op}}$ requires $\sim\!36\times$ more
time-steps for matched bias, and the empirical $64\times$ refinement
required to recover convergence is within a factor of two of the
prediction.

\paragraph{Solver choice for the experiments in this paper.}
We use ExpSplit for all training and evaluation runs reported in the
body of the paper, except for the diamond system, which uses
Euler--Maruyama. The amplitude-sensitivity experiment of
Section~\ref{app:scaling:em_amp} shows that EM accuracy at fixed
$n_{\mathrm{time\text{-}steps}}$ degrades quickly once peak
$\|H_{\mathrm{eff}}\|$ rises above $\sim\!30$, requiring roughly
$\|H\|^2$-scaled refinement of the time grid to recover convergence.
ExpSplit's deterministic half is Hamiltonian-exact and removes this sensitivity \rev{for the Hamiltonian part}, so we do not have to revisit the time-grid
resolution when sweeping decoherence rates, control amplitudes, or
the number of qubits. 

For the diamond system
EM is within its accurate regime: the four-level Hilbert
space and the small-amplitude controls ($\lambda_{\mathrm{flu}}=
0.001$ keeps the controller amplitude well below the threshold of
Table~\ref{tab:em_accuracy_amp}) place it \rev{in the converged regime}, and at $\dim=4$ the $O(\dim^2)$ trajectory representation is
strictly faster than ExpSplit's $O(\dim^3)$ matrix exponential. The
larger benchmarks (chain-4 at $\dim=16$, chain-6 at $\dim=64$, plus
the amplitude-damping and STIRAP sweeps\rev{, as a precaution for the strongest-noise settings}) use ExpSplit.
For systems beyond the dimensions in this paper EM remains the
appropriate choice, with $n_{\mathrm{time\text{-}steps}}$ chosen to
match the expected operator norm.

\section{Compute budget}
\label{app:compute}
All experiments ran on a single NVIDIA A100 GPU (40\,GB,
\texttt{gpu\_a100}). The compute consumed by the experiments reported
in this paper breaks down approximately as $\sim$55\,h on
amplitude-damping, STIRAP, and diamond system training (including
their PPO baselines); $\sim$25\,h on the chain-4 $\rho$-sweep;
$\sim$30\,h on the IBM Kingston q14$\to$q9 chain (including the PPO
hyperparameter sweep and multi-seed extension); and $\sim$10\,h on
the scaling benchmarks of Appendix~\ref{app:scaling}, for a reported
total of approximately $\mathbf{120}$\,\textbf{A100-hours}. The full
project, including hyperparameter sweeps, alternative system
parameterizations, and exploratory variants that did not make it into
the final paper, consumed approximately $\mathbf{240}$\,\textbf{A100-hours}
of GPU compute.

\end{document}